\documentclass[11pt]{amsart}

\usepackage{eufrak}
\usepackage{amsfonts}
\usepackage{float}
\usepackage{amsmath,amssymb}
\usepackage{lscape}
\usepackage{graphicx, color}

\usepackage{dashbox}
\usepackage{algpseudocode}
\usepackage{algorithm}
\usepackage{enumerate}

\newcommand{\R}{\mathbb{R}}

\newtheorem{theorem}{\bf Theorem}[section]
\newtheorem{definition}[theorem]{\bf Definition}
\newtheorem{remark}[theorem]{\bf Remark}
\newtheorem{lemma}[theorem]{\bf Lemma}

\newcommand{\Face}{\mbox{Face}}
\newcommand{\Center}{\mbox{Center}}
\newcommand{\Edge}{\mbox{Edge}}
\newcommand{\Extent}{\mbox{Extent}}
\newcommand{\Prev}{\mbox{Prev}}
\newcommand{\Next}{\mbox{Next}}
\newcommand{\Ridge}{\mbox{Ridge}}
\newcommand{\Parent}{\mbox{Parent}}
\newcommand{\Forecast}{\mbox{Forecast}}
\newcommand{\IsPropagated}{\mbox{IsPropagated}}

\begin{document}

\title[Cut-locus on a polyhedron]
{An extended MMP algorithm: \\
	wavefront and cut-locus on a convex polyhedron}

\author[K.~Tateiri]{Kazuma Tateiri}\address[K.~Tateiri]{Graduate School of Information Science and Technology, Hokkaido University,Sapporo 060-0810, Japan}\email{kazuma.tateiri16@gmail.com}

\author[T.~Ohmoto]{Toru Ohmoto}\address[T.~Ohmoto]{Department of Mathematics, Hokkaido University,Sapporo 060-0810, Japan}\email{ohmoto@math.sci.hokudai.ac.jp}

\subjclass[2010]{}

\keywords{}

\dedicatory{}

\begin{abstract}
In the present paper, we propose a novel generalization of the celebrated MMP algorithm in order to find the {\em wavefront propagation} and the {\em cut-locus} on a convex polyhedron with an emphasis on actual implementation for instantaneous visualization and numerical computation.
\keywords{geodesics; wavefront propagation; cut locus; source unfolding.}
\end{abstract}

\maketitle

\section{Introduction}
Geometry of geodesics on polyhedra is very rich -- it attracts people since ancient times.
Nowadays, finding shortest paths and shortest distances has many real applications in engineering and industrial fields, and several algorithms for computing them have been proposed so far, see \cite{DO,Crane,Bose} and references therein.
Among polyhedral approaches, most well known is the so-called MMP algorithm, given by Mitchell, Mount and Papadimitriou \cite{MMP} (also \cite{Mount}).
We revisit this classical and well established method in computational geometry.
The aim of the present paper is to propose a novel generalization of the MMP algorithm for finding some richer structure of geodesics, the {\em wavefront propagation} and the {\em cut-locus}. While our method has some limitations (discussed later), we emphasize our actual implementation to computer program for instantaneous visualization and numerical computation \footnote{The source code is available at https://github.com/Raysphere24/IntervalWavefront}.

As a toy example, look at Figure \ref{bunny}. Here we take the convex hull of {\em Stanford Bunny} (the left picture, viewed transparently).
On this polyhedron, choose freely a source point indicated by $\times$ colored by yellow, then our algorithm creates the time-evolution of wavefronts (yellow curve in the middle) instantaneously and accurately enough, and finally it ends at the right picture. Red dots represent ridge points on the wavefront curve. As the time increases, the ridge points sweep out the cut-locus colored by green. In Figure \ref{bunny_ws}, the wavefront propagation is observed from different viewpoints, and in Figure \ref{bunny_opaque} the cut locus is viewed opaquely.

\begin{figure}[h]
	\centering
	\includegraphics[height=4cm, pagebox=cropbox]{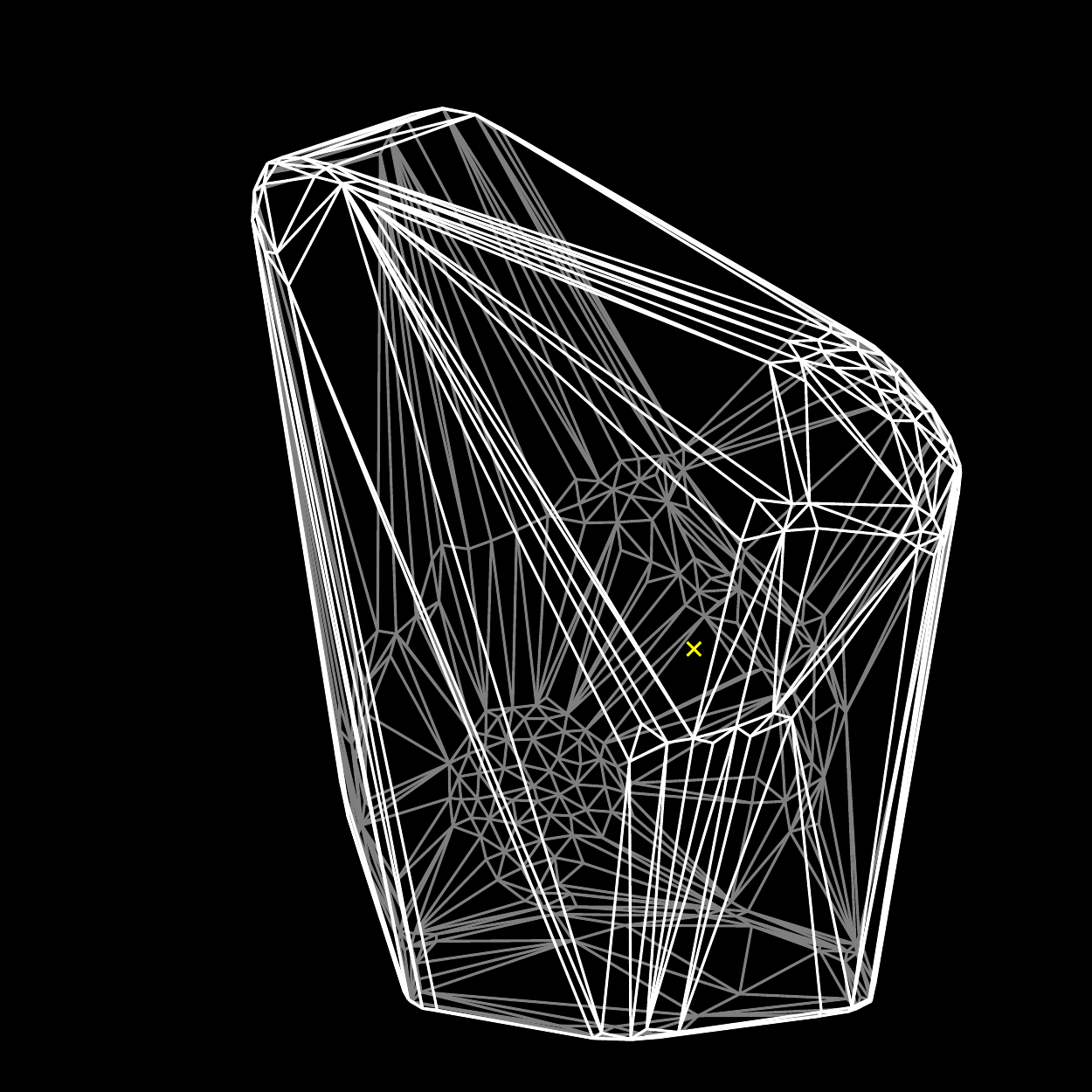}
	\includegraphics[height=4cm, pagebox=cropbox]{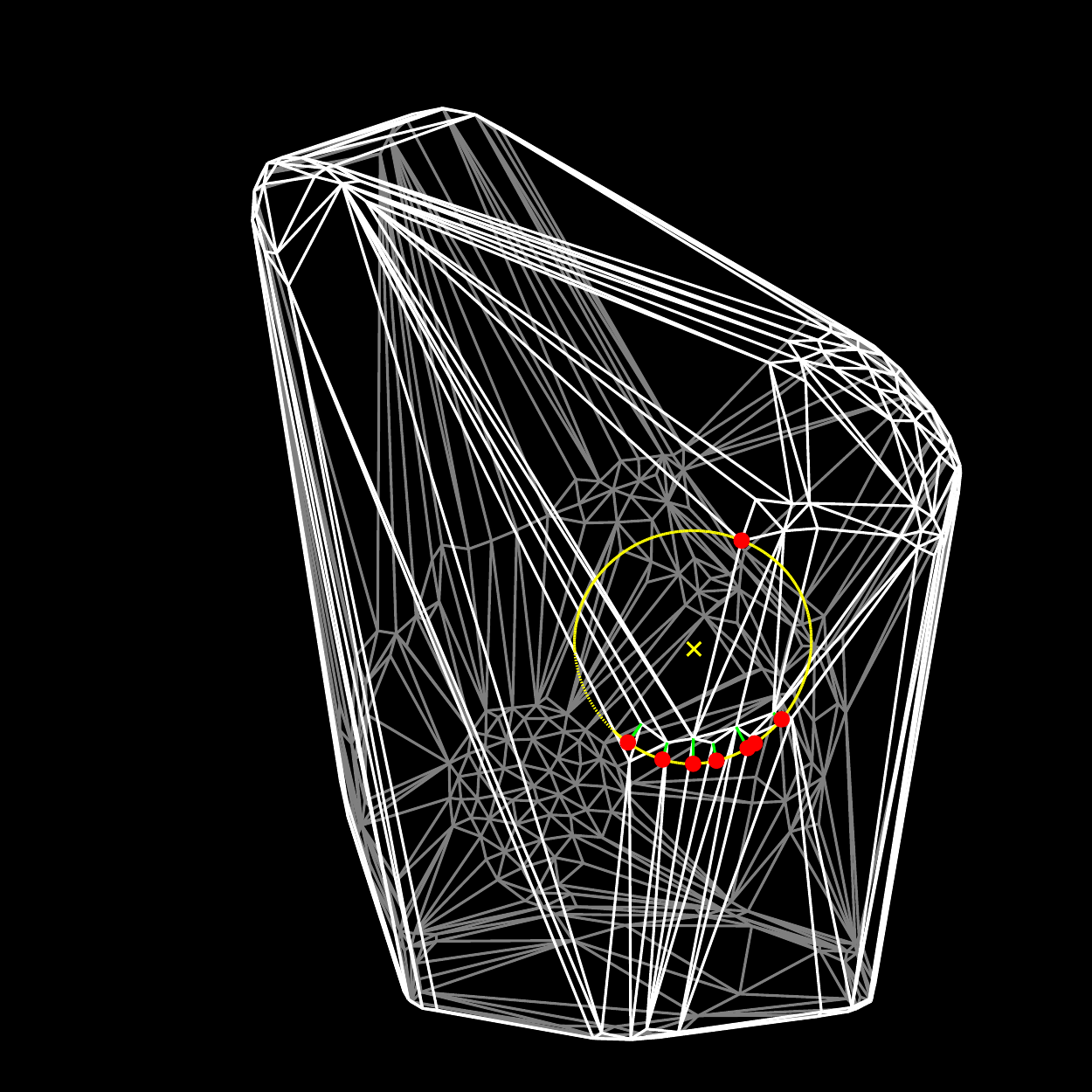}
	\includegraphics[height=4cm, pagebox=cropbox]{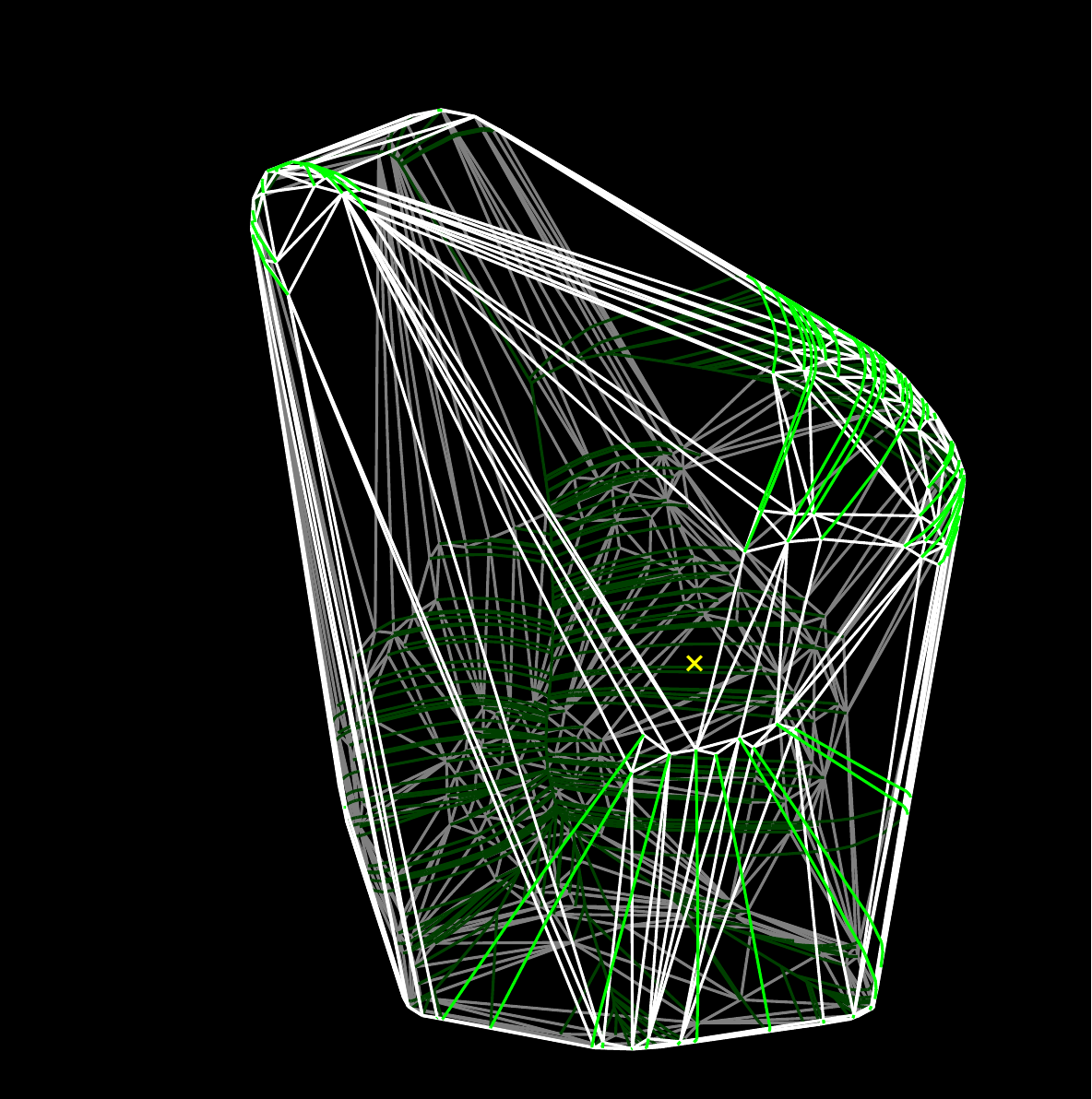}
	\caption{\small Wavefront and cut-locus.}
	\label{bunny}
\end{figure}

\begin{figure}[h]
	\centering
	\includegraphics[height=3cm, pagebox=cropbox]{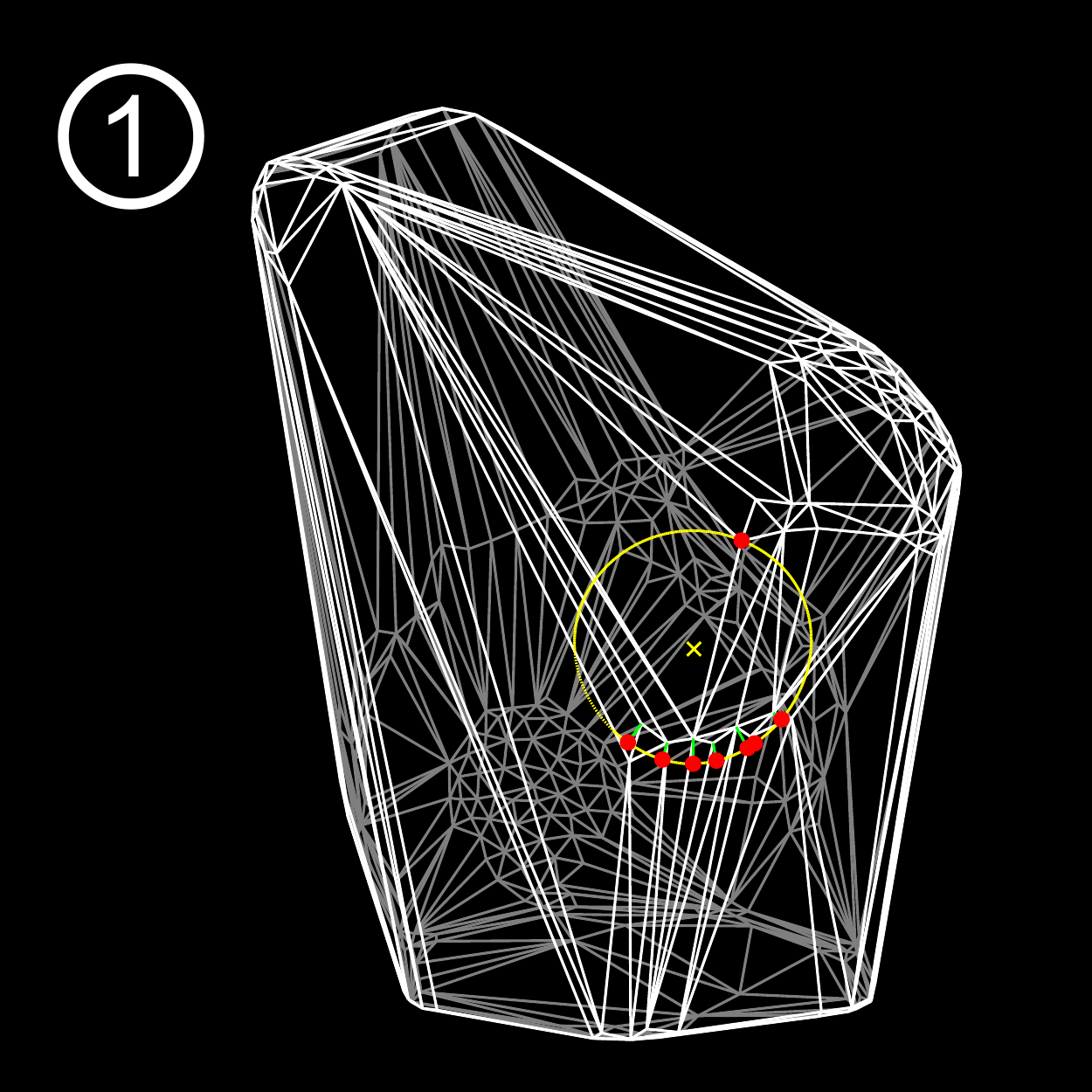}
	\includegraphics[height=3cm, pagebox=cropbox]{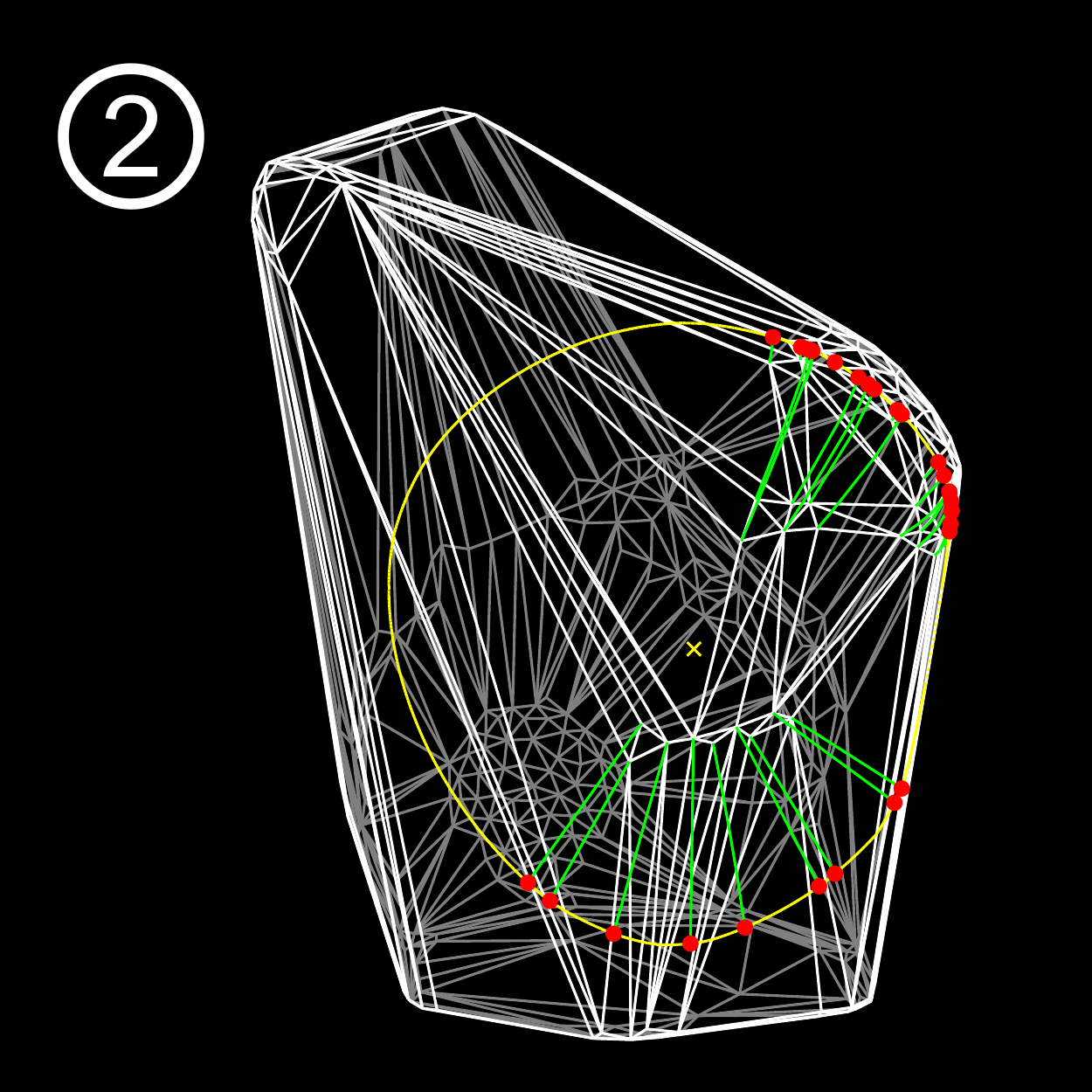}
	\includegraphics[height=3cm, pagebox=cropbox]{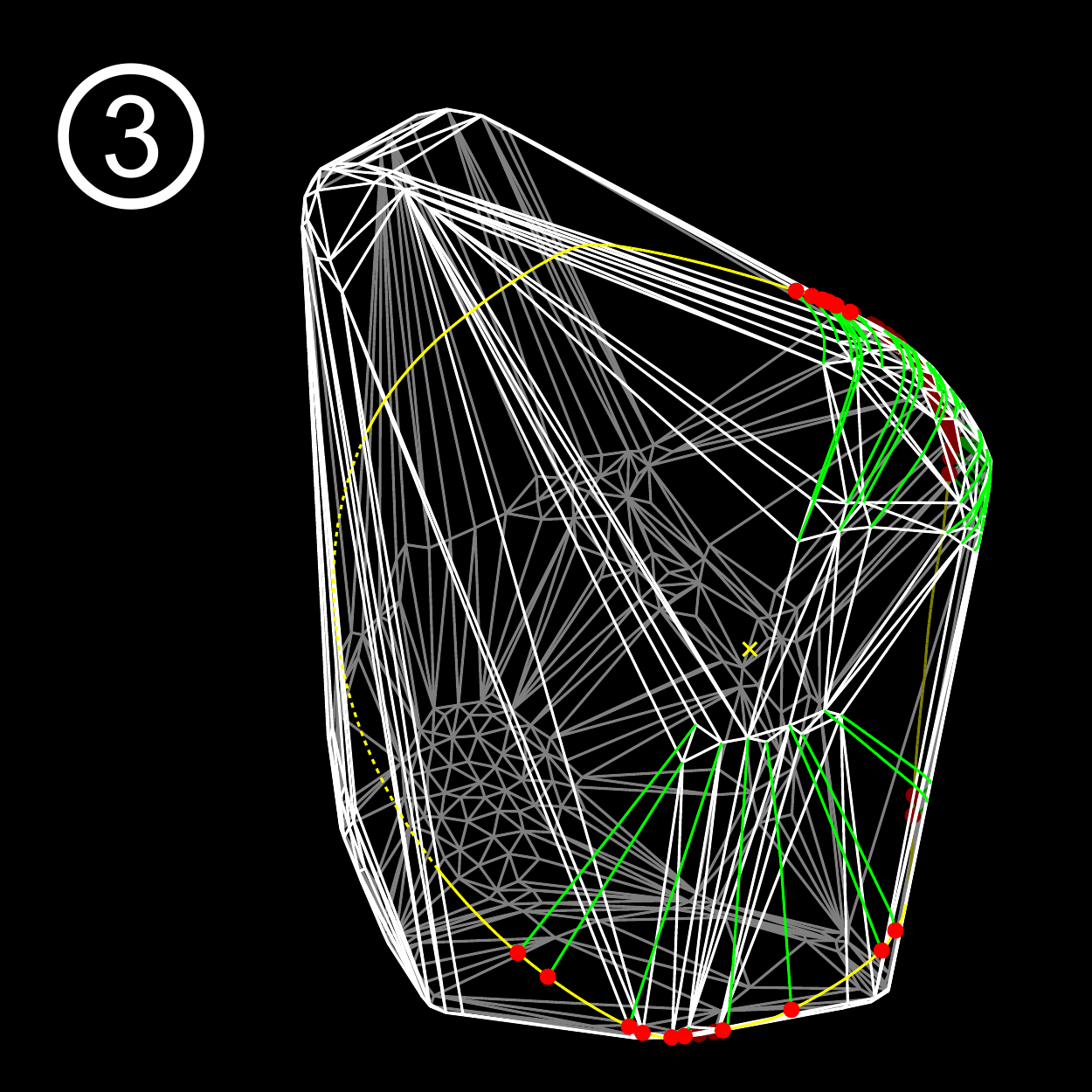}
	\includegraphics[height=3cm, pagebox=cropbox]{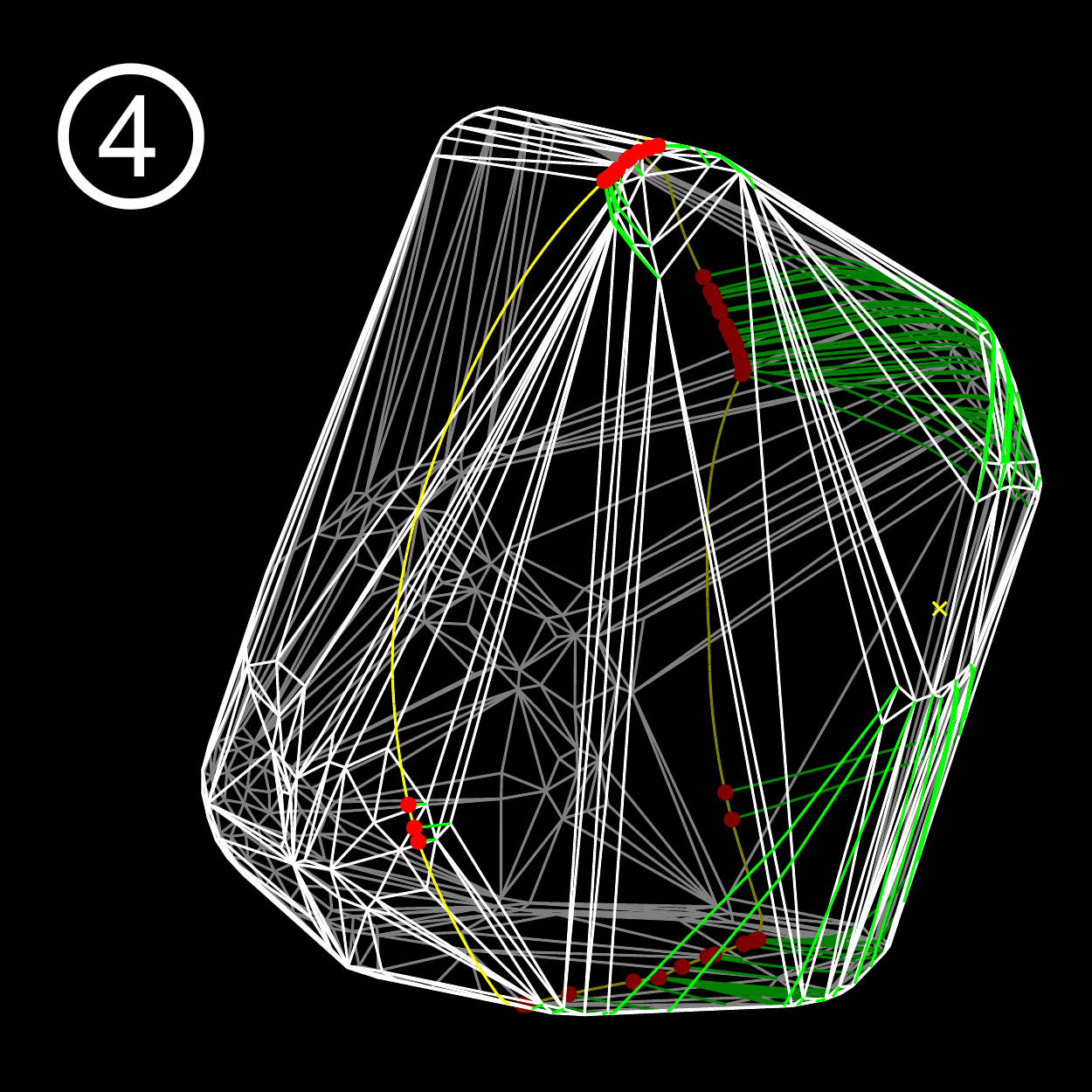}
	\includegraphics[height=3cm, pagebox=cropbox]{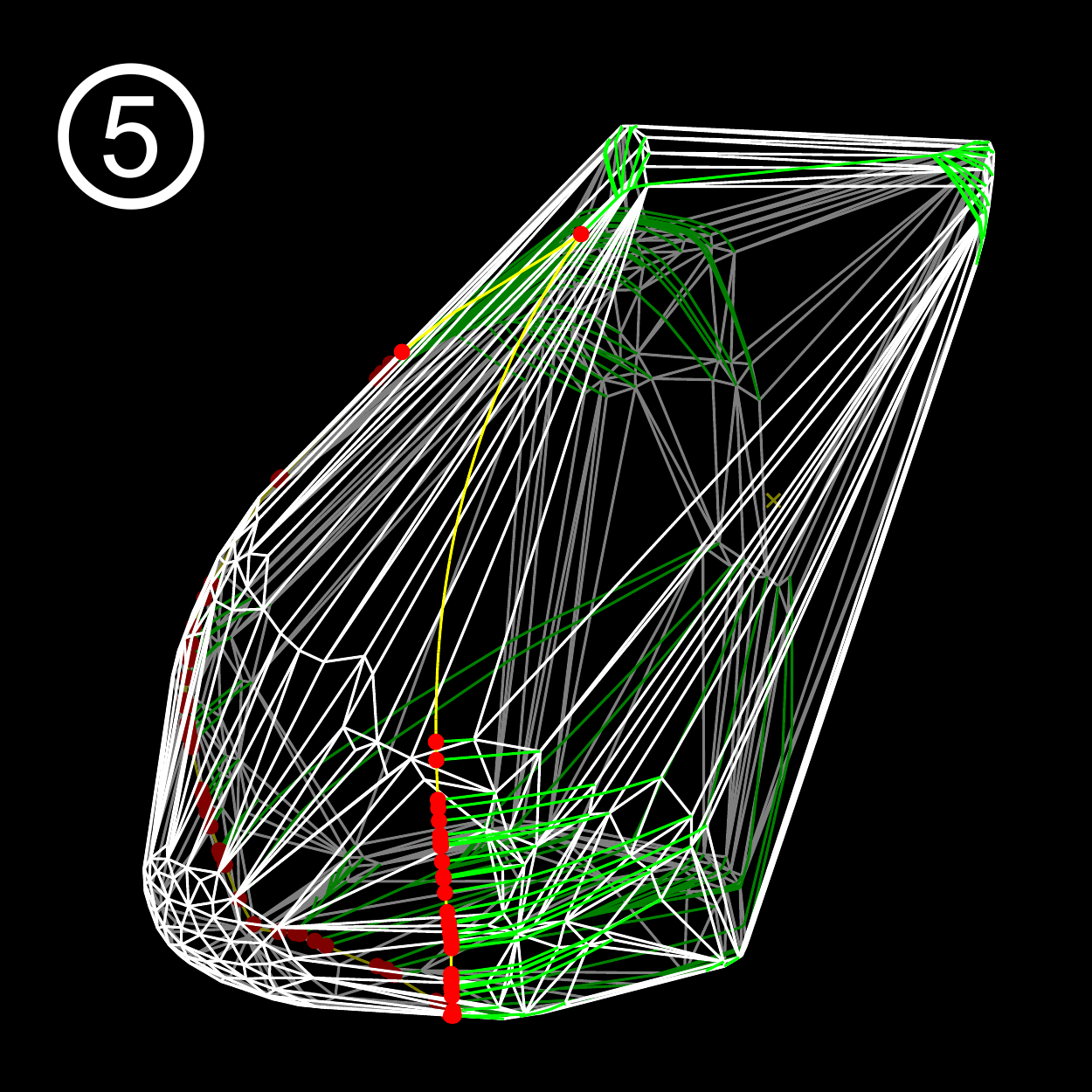}
	\includegraphics[height=3cm, pagebox=cropbox]{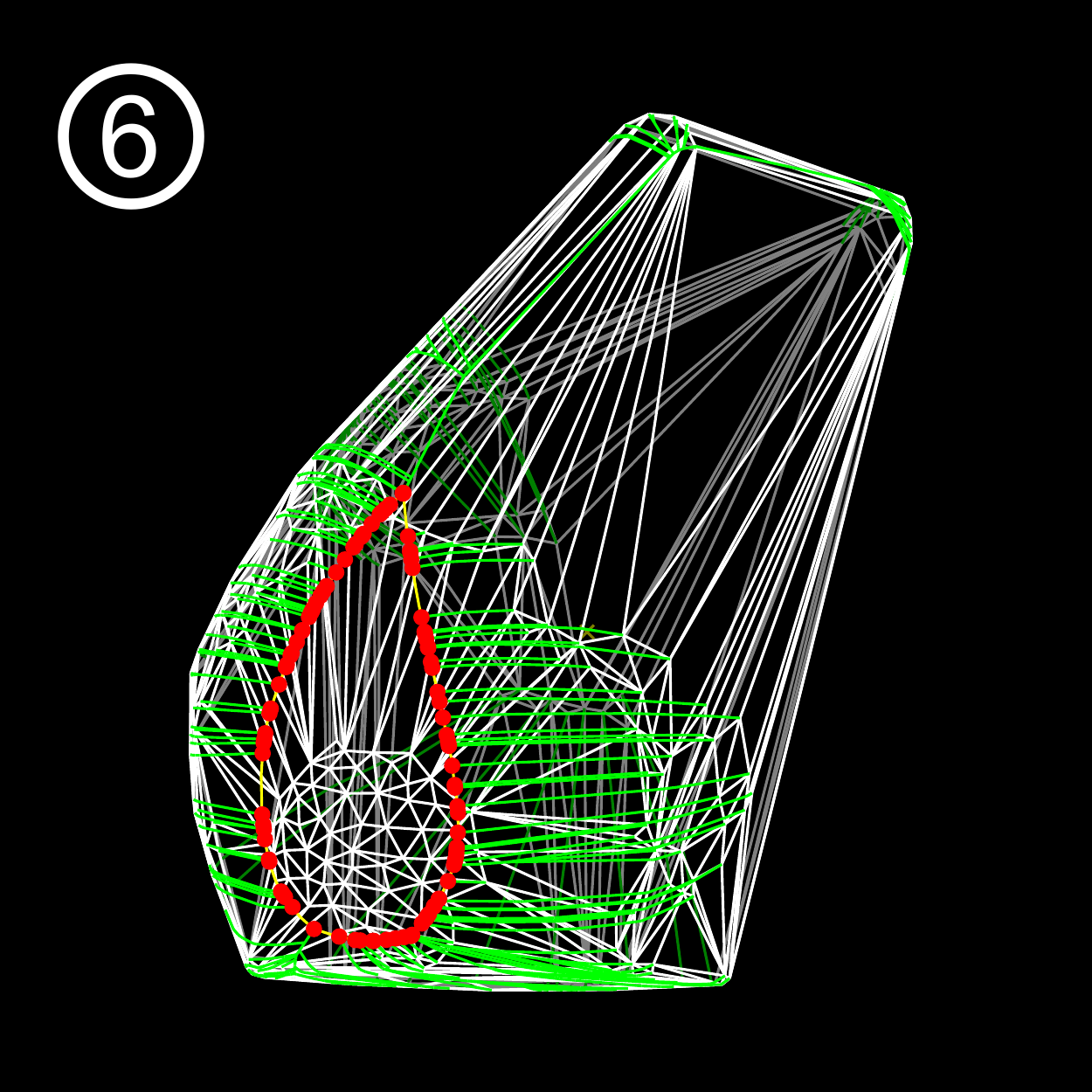}
	\includegraphics[height=3cm, pagebox=cropbox]{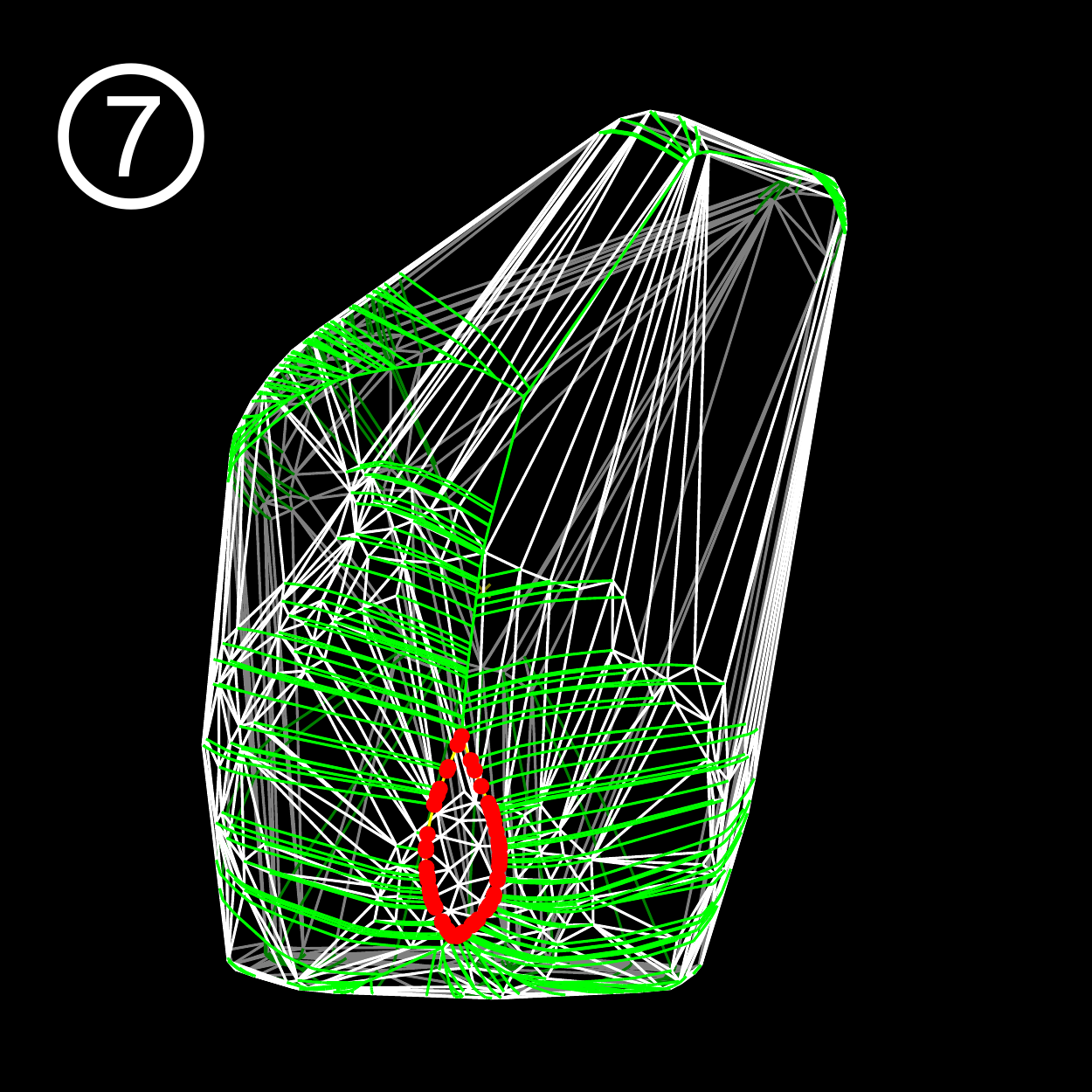}
	\includegraphics[height=3cm, pagebox=cropbox]{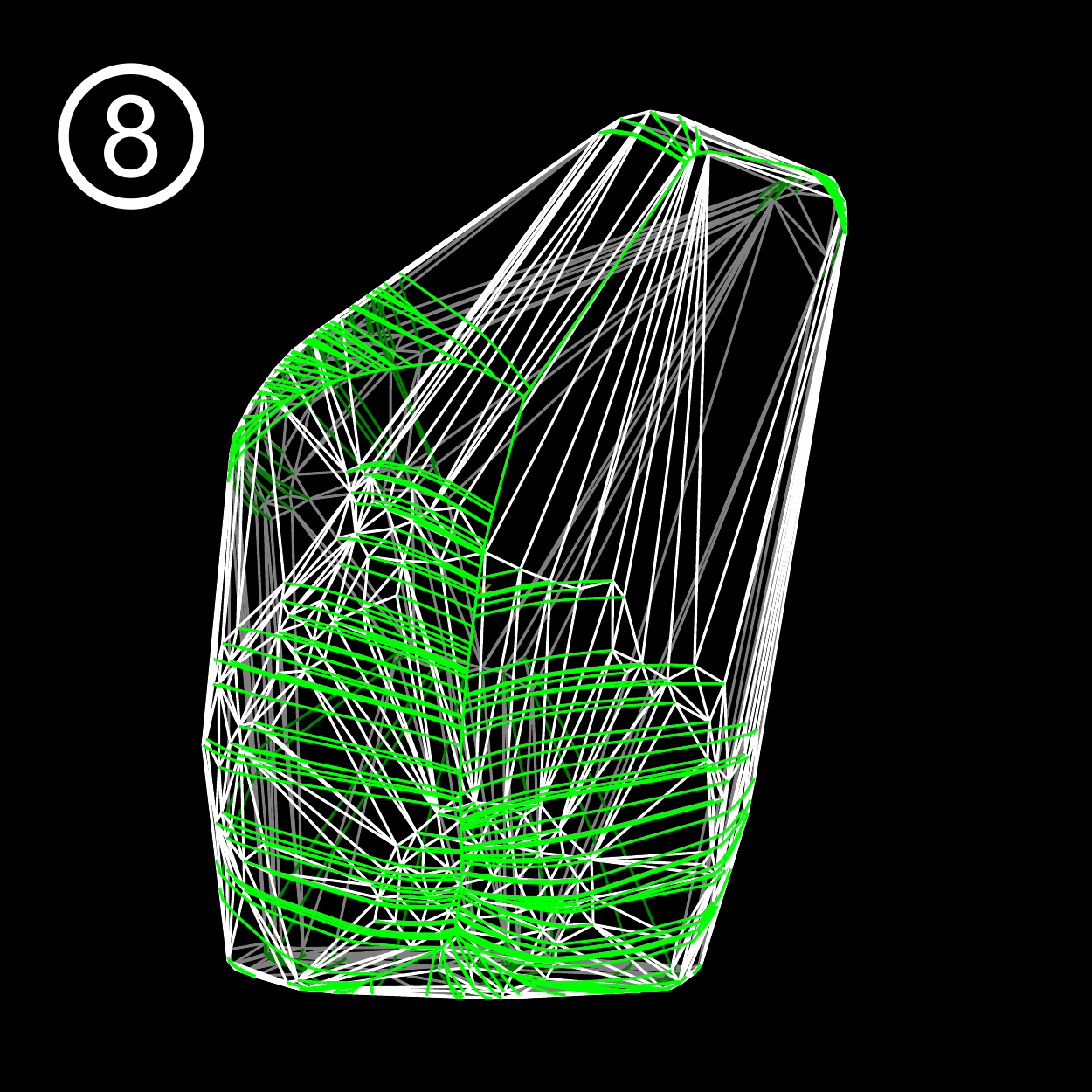}
	\caption{The wavefront $W(r)$ (yellow curve) propagates from no.1 to no.8. In our program,
	the viewpoint can be chosen freely in an interactive way;
	no.8 is a different view of the right picture in Figure \ref{bunny}.
	}
	\label{bunny_ws}
\end{figure}

\begin{figure}[h]
	\centering
	\includegraphics[height=6cm, pagebox=cropbox]{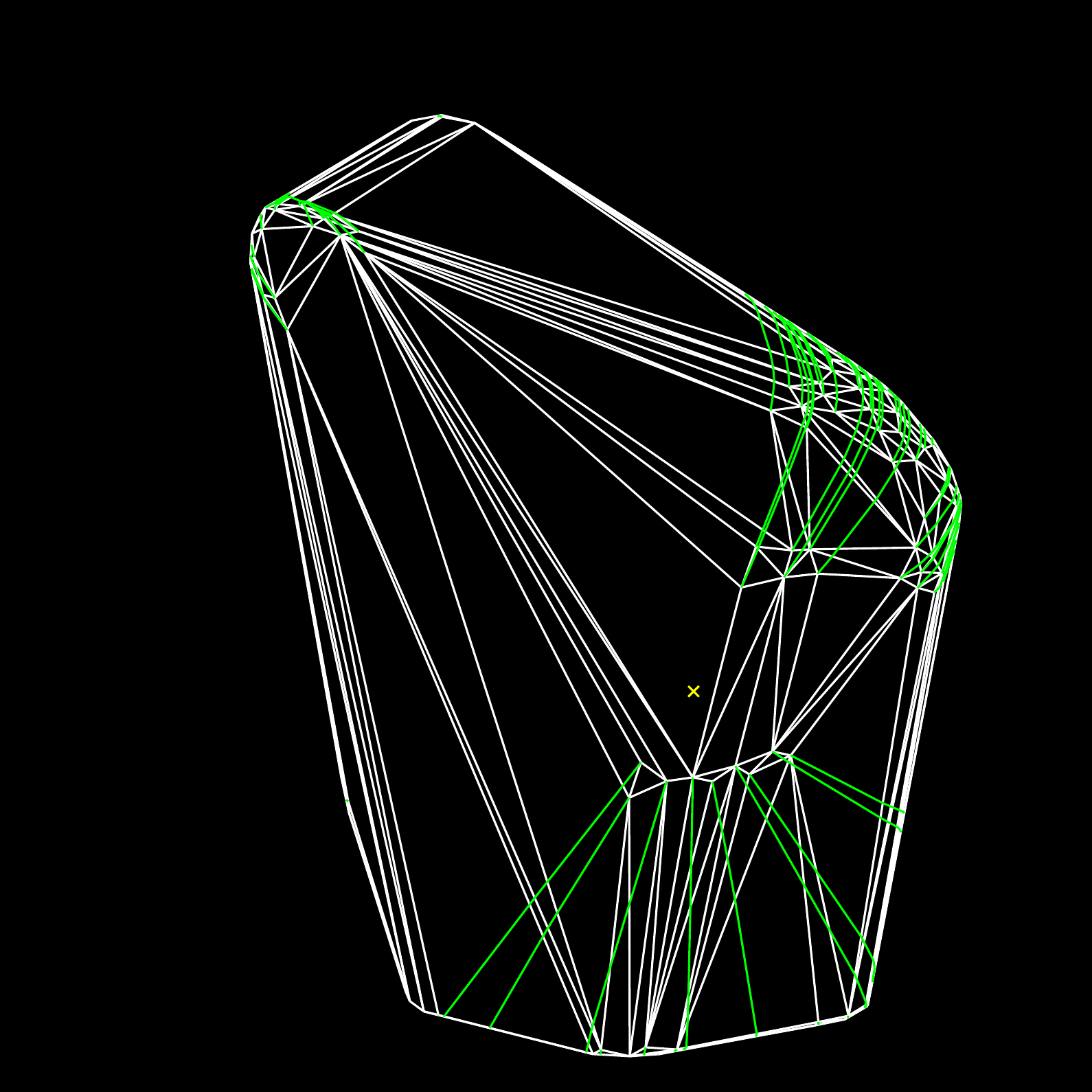}
	\includegraphics[height=6cm, pagebox=cropbox]{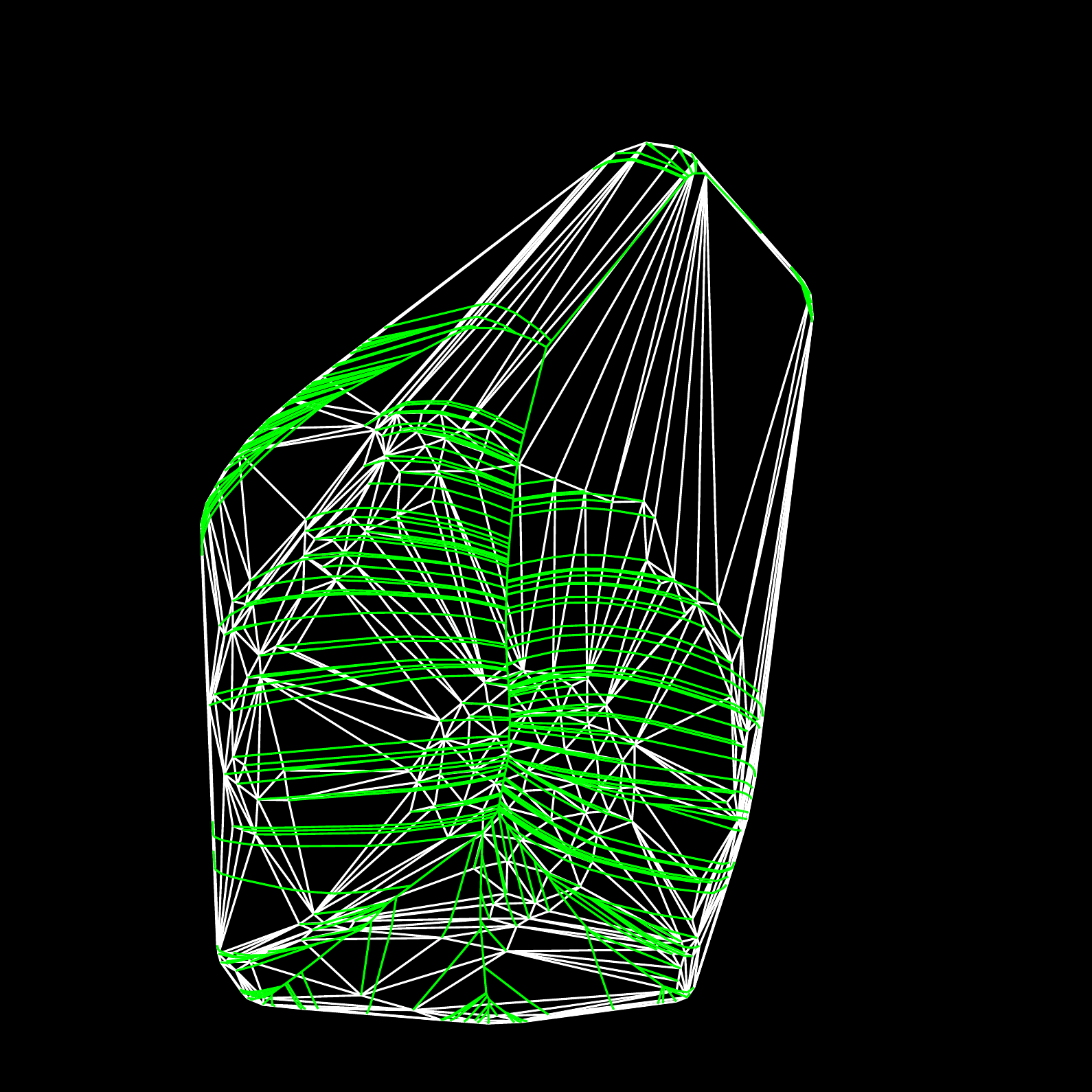}
	\caption{Opaque renderings of the cut-locus (two different viewpoints)}
	\label{bunny_opaque}
\end{figure}

First, we fix basic terminologies precisely.
Let $S$ be a convex polyhedral surface in Euclidean space $\R^3$, and let $d:S\times S \to \R$ denote the distance function on $S$. Pick a point $p$ of $S$, and call it a {\em source point}.
Given $r>0$, the {\em wavefront} on $S$ caused from $p$ is defined by the set of points of $S$ with iso-distance $r$ from $p$:
$$W(r):=\{\; x \in S \; | \; d(p, x) =r\; \}$$
(also we may write it by $W(r, p)$). Suppose that the wavefront $W(r)$ propagates on $S$ with a constant speed, as $r$ varies. If $p$ is an interior point of a face, then initially $W(r)$ is just a circle centered at $p$ with small radius $r$ on the face. As $r$ increases, $W(r)$ is still a loop on $S$, until it collapses to the farthest point from $p$ (Figure \ref{bunny_ws}) or it occurs a self-intersection and breaks off into multiple disjoint pieces (we call the moment a {\em bifurcation event}, see Figure \ref{genericevents} (b)). Here our method has a limitation of incapable of cope with a bifurcation event, and in that case, we are interested in the wavefront propagation up to that moment. We will discuss this point later.

Geometrically, $W(r)$ is made up of circular arcs.
Two neighboring arcs may be joined by a {\em ridge point} of $W(r)$, which is a point having at least two distinct shortest paths from $p$.
A new ridge point is created when $W(r)$ hits a vertex of $S$ (we call it a {\em vertex event}), see Figure \ref{ridge_point}.
The locus of ridge points of $W(r)$ for all $r>0$ is called the {\em cut-locus} $C=C(S, p)$ on $S$ (here let $C$ contain all vertices of $S$, at which vertex events happen).
It is a graph embedded in $S$ and has a clear geometric meaning.
When one cuts $S$ along $C$ by a scissor, the whole of $S$ is expanded to the plane so that the obtained unfolding (net, development) is a star-shaped polygon without any overlap -- every point of the unfolding can be joined with $p$ by a line segment lying on it (Figure \ref{source_unfolding}), which corresponds to a shortest path on $S$. We call it the {\em source unfolding of $S$ centered at $p$}.

\begin{figure}[h]
	\centering
	\includegraphics[height=6cm, pagebox=cropbox]{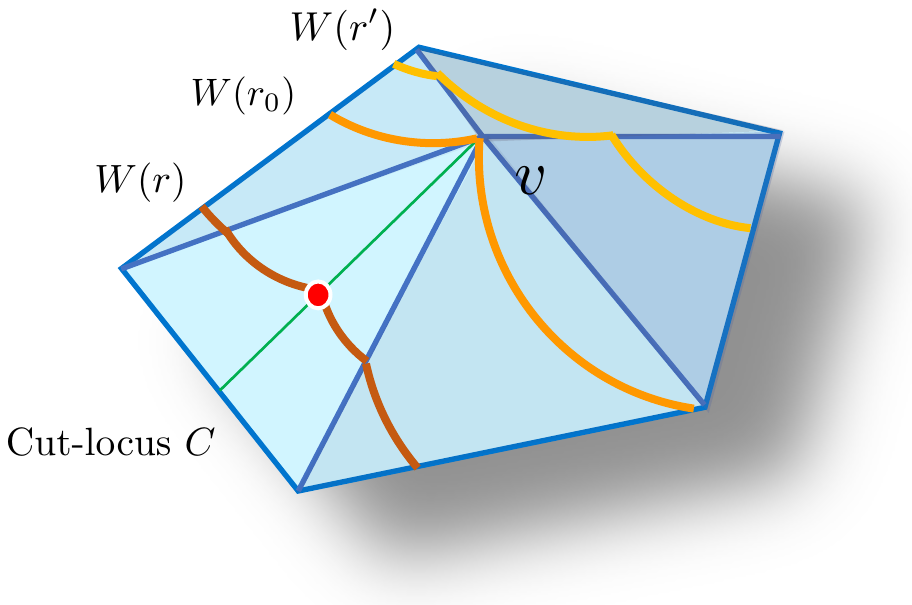}
	\caption{Ridge points are born at vertices and sweep out the cut-locus ($r'<r_0<r$).}
	\label{ridge_point}
\end{figure}

\begin{figure}[h]
	\centering
	\includegraphics[width=9.5cm, pagebox=cropbox]{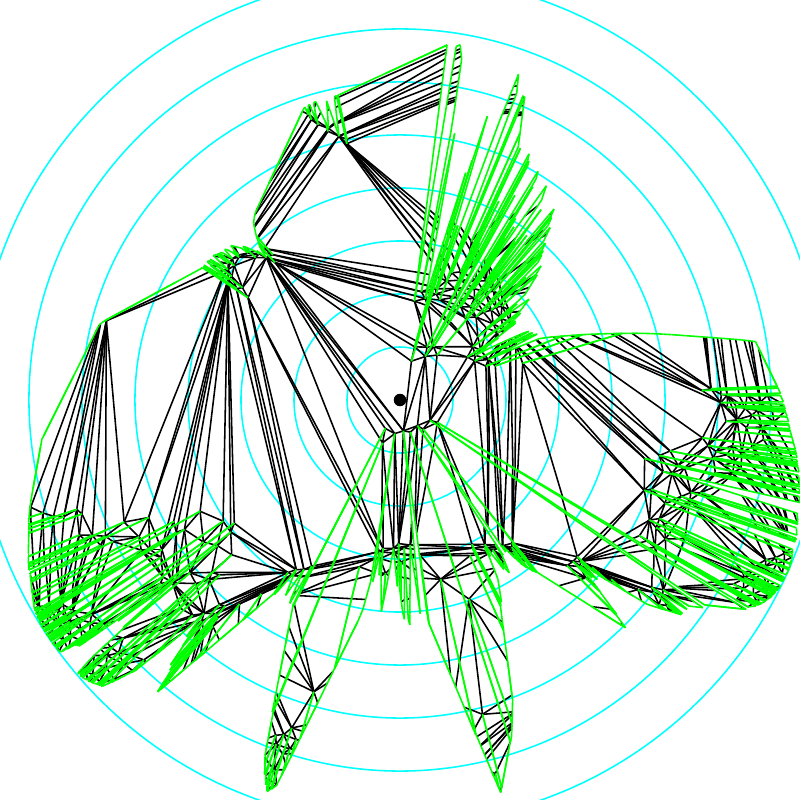}
	\caption{Source unfolding for the same example as in Figures \ref{bunny} and \ref{bunny_ws}; our program immediately produces it just after the wavefront propagation ends. Concentric circles on the unfolding represent the wavefronts on $S$.
	}
	\label{source_unfolding}
\end{figure}

Research in computational geometry on the cut-locus and its source unfolding has been investigated so far by several authors, e.g., \cite{CH,DO,Itoh,MMP,Mount}. Nevertheless, our approach seems to be new.
Our problem is to interactively visualize the wavefront propagation $W(r)$ and compute the cut-locus $C$ as $r$ varies, and to finally produce the source unfolding of $S$ precisely.
That is designed for a practical and interactive use -- for instance, in our specifications, the source point $p$ is chosen by a click on the screen and the viewpoint for $S$ can freely be rotated manually.
Here is a key point that we may assume that $p$ lies in sufficiently general position; this practical assumption enables us to classify geometric events arising in the propagation into several types (see \S 2), and then the algorithm becomes simple enough to be treated.
Actually, our computer program certainly works, even when we choose $p$ lying on an edge or a vertex in a visible sense (i.e., choose $p$ within a very small distance $\varepsilon (\ll 1)$ from the edge or the vertex), see Figure \ref{icosahedron_vertex_or_edge}.

\begin{figure}[h]
	\centering
	\includegraphics[height=3.75cm, pagebox=cropbox]{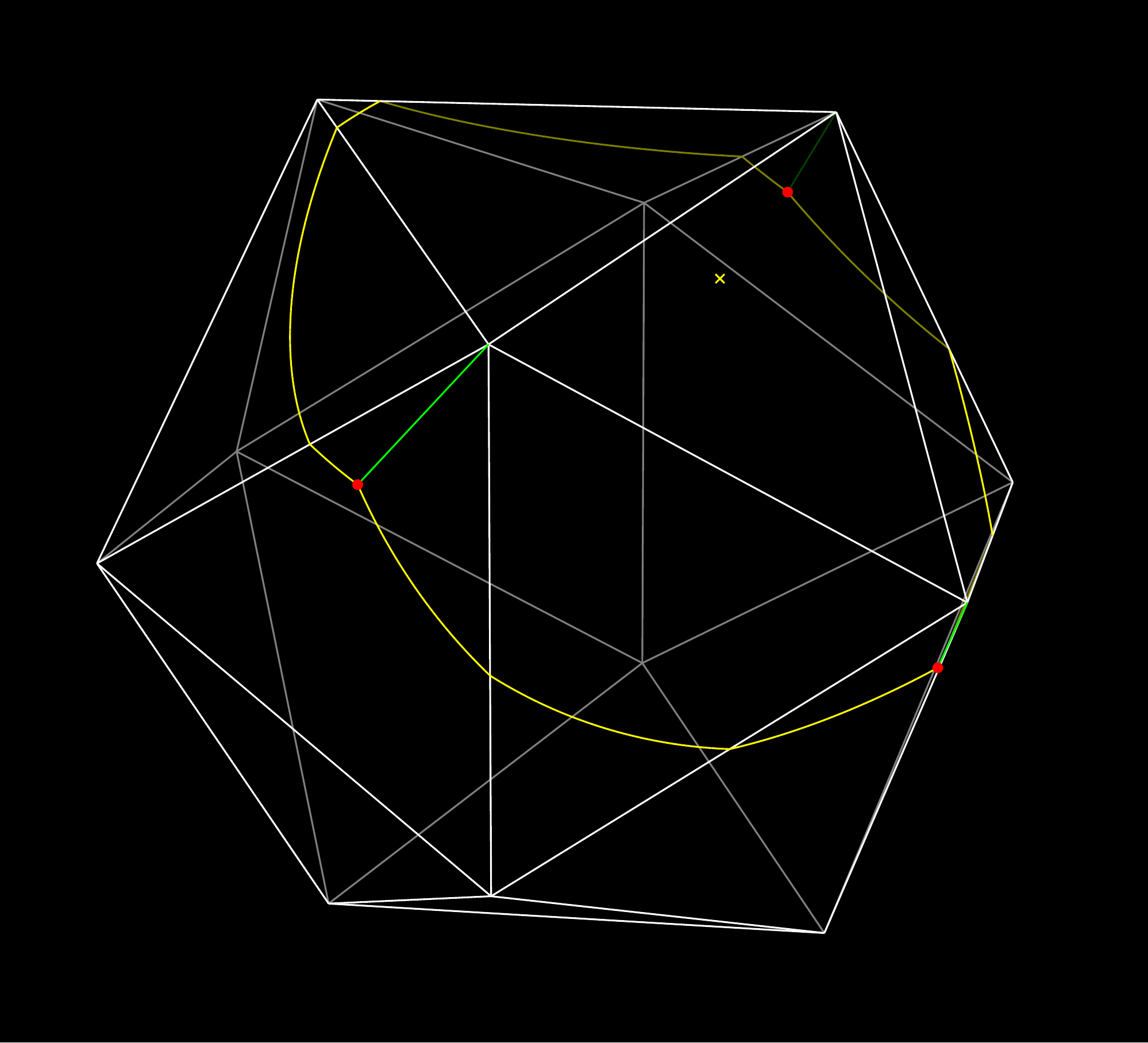}
	\includegraphics[height=3.75cm, pagebox=cropbox]{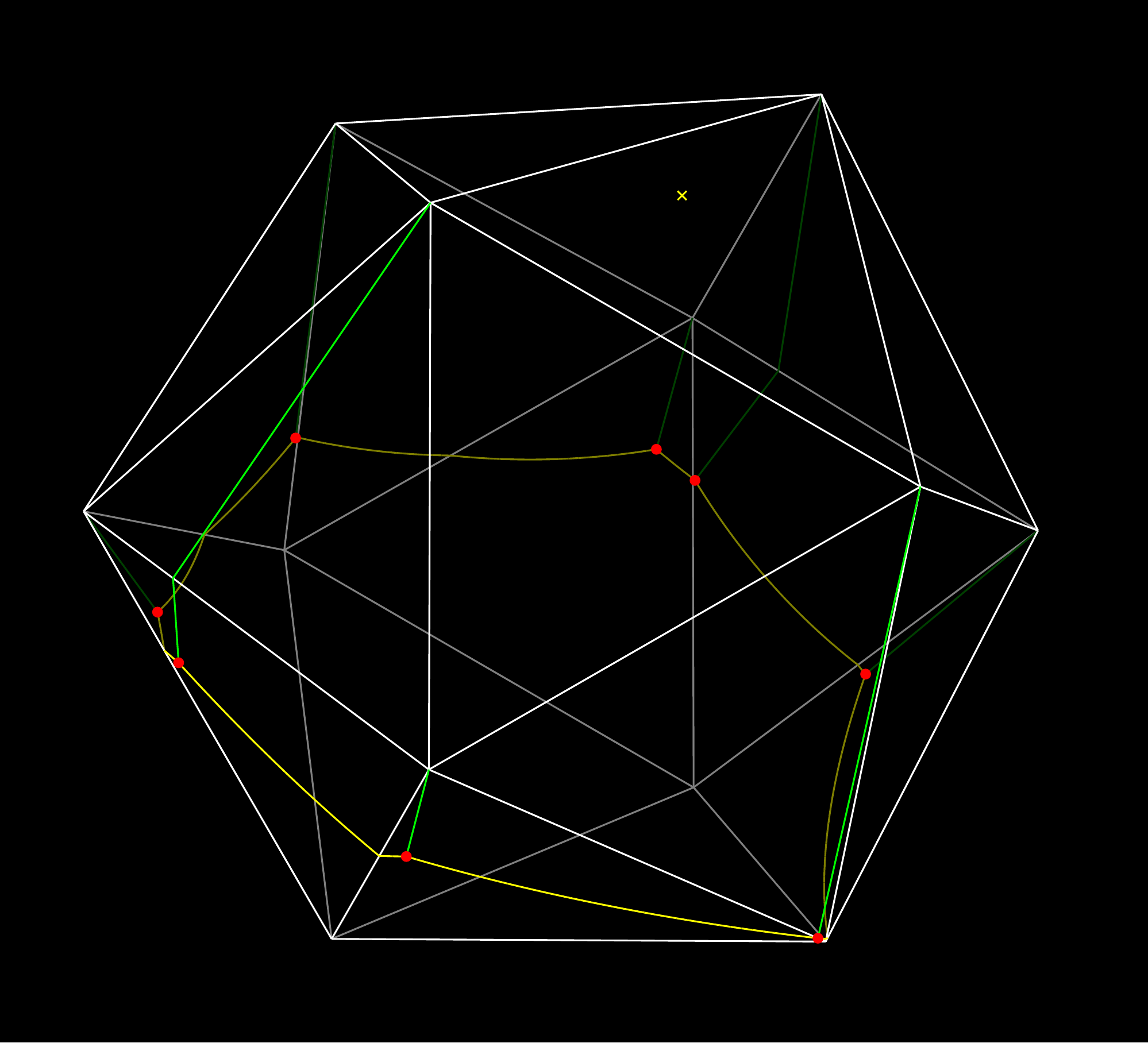}
	\includegraphics[height=3.75cm, pagebox=cropbox]{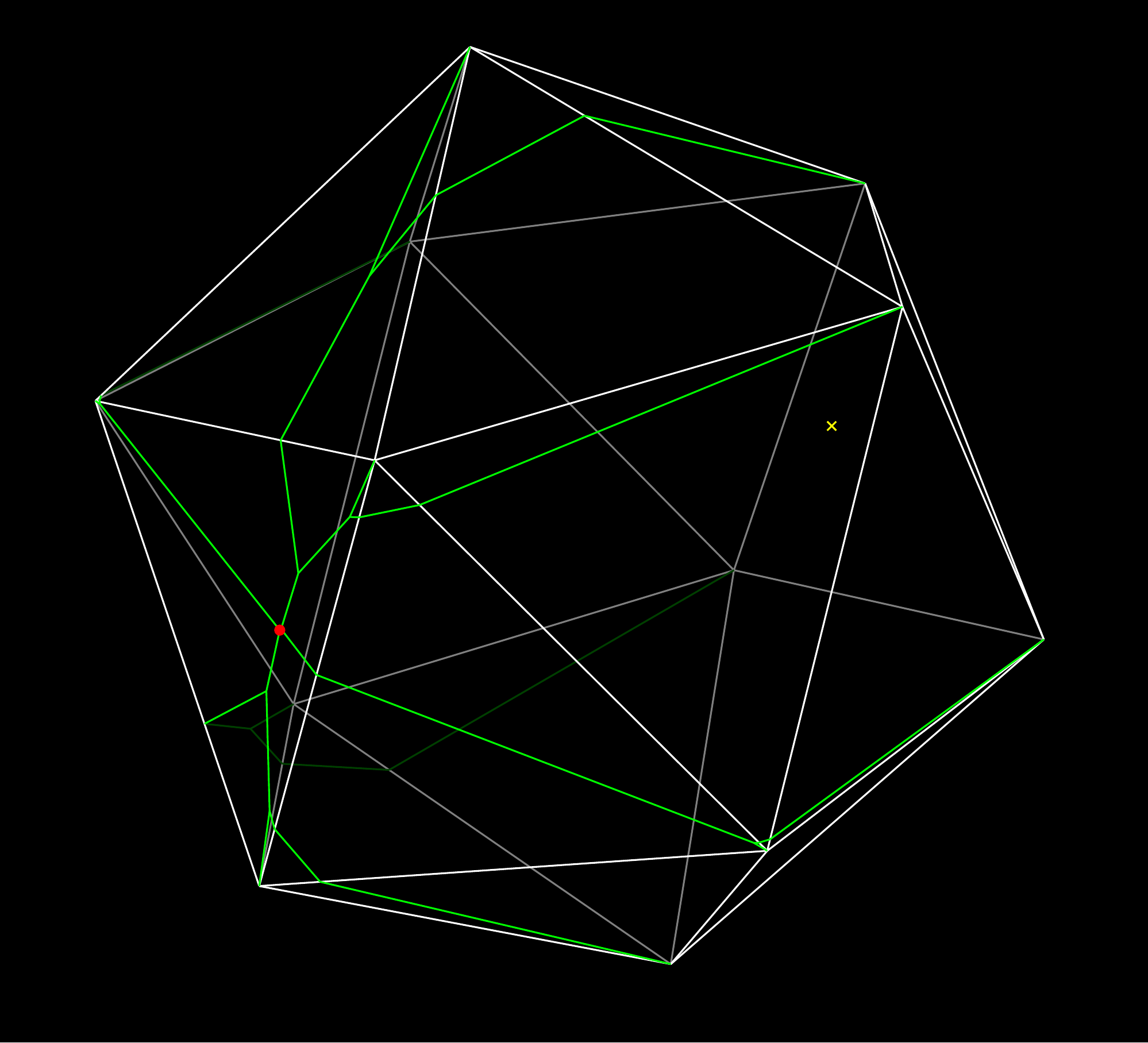}\\
	\vspace{8pt}

	\includegraphics[height=6cm, pagebox=cropbox]{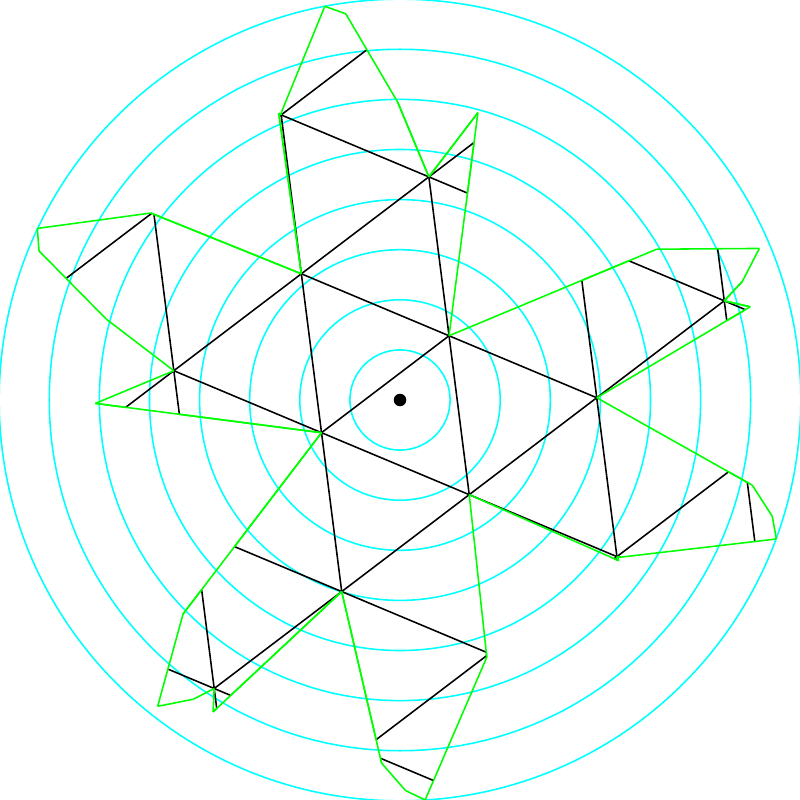}
	\caption{Source unfolding of an icosahedron. One can create a number of examples by choosing different source points.}
\end{figure}

\begin{figure}[h]
	\centering
	\includegraphics[height=5cm, pagebox=cropbox]{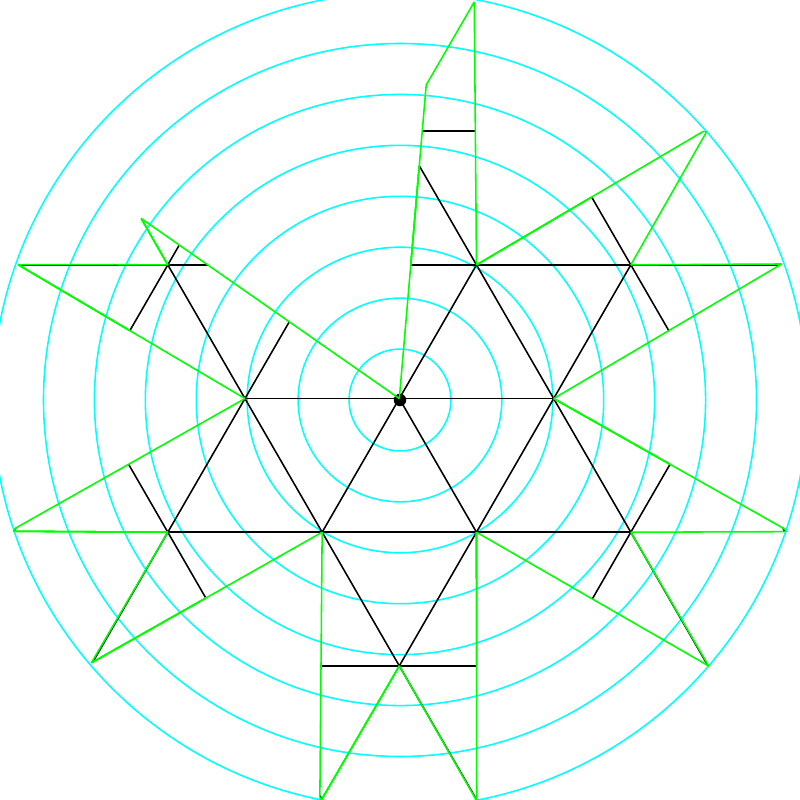}\; \;
	\includegraphics[height=5cm, pagebox=cropbox]{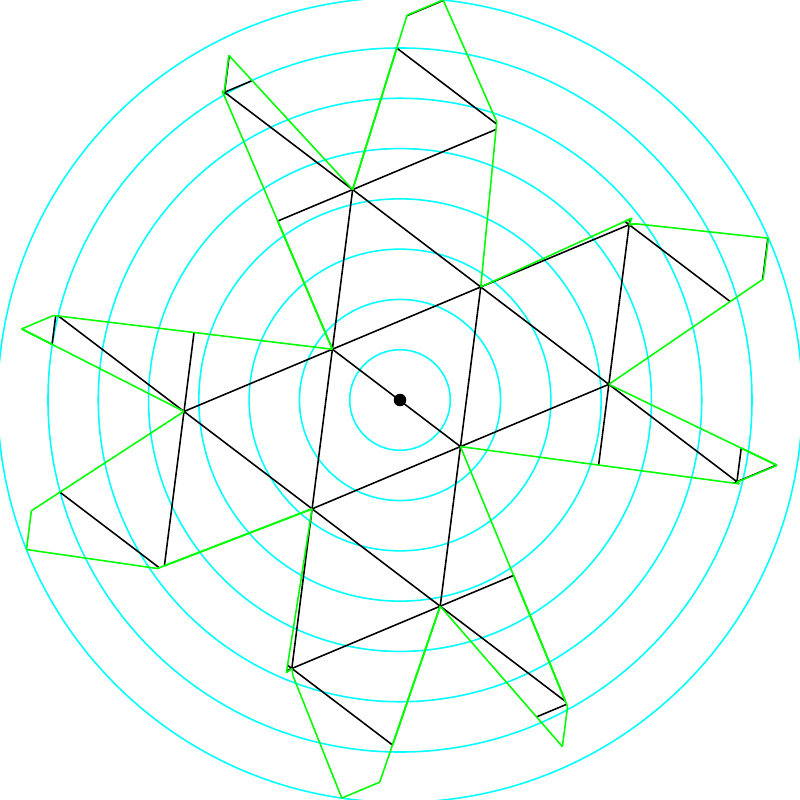}
	\caption{Different source unfoldings. Manually selecting the source point $p$ on a vertex or edge contains a small invisible error. This enables our algorithm to work and properly create the cut-locus.}
	\label{icosahedron_vertex_or_edge}
\end{figure}

The MMP algorithm aims to compute the shortest geodesic between two points; it receives $S$ and $p$ as inputs, and results a specialized data structure called {\em intervals}, which are subdivisions of all edges of $S$ equipped with some additional information.
However, the MMP algorithm itself and other existing algorithms are insufficient for our practical purpose.
We try to improve the MMP algorithm -- one of our major ideas is to introduce a new data structure, called an {\em interval loop} indexed by the parameter $r$, which is a recursive sequence of {\em enriched intervals} $I_i\, (=I_i^{(r)})$
$$\mathbf{I}_r: I_0 - I_1 - I_2 -\cdots - I_{k(r)} - I_0.$$
Roughly speaking, $\mathbf{I}_r$ is the data structure representing the wavefront $W(r)$;
each enriched interval corresponds to a circular arc participating in $W(r)$, and the sequence is closed, because $W(r)$ is assumed to be an oriented closed curve embedded in $S$. Again, since our method disallows a bifurcation event, $W(r)$ is connected and one interval loop represents all of $W(r)$.

As $r$ increases, there arise some particular moments $r_0$, e.g., a new ridge point is born from a vertex as depicted in Figure \ref{ridge_point}, a ridge point meets an edge, two or more ridge points collide on a face, and so on. We call them {\em geometric events} of the wavefront propagation.

Here we make a simplification to this event model. An interior point of an arc often reaches an edge before either of the surrounding two ridge points meets the edge, and then the arc is pushed out to the next face (or the next of the next, or so on), while the two ridge points stay on the original face. Since we are primarily interested in the ridge points, we do not recognize this phenomenon as an event, while the interval loop does not precisely represent the wavefront. This definition makes our algorithm much simpler, as we explain later (see Remark \ref{partial_propagation}). This simplification could also improve practical performance by a constant factor, which is hidden behind the big-O notation.

Until a new event occurs at the time $r_0\, (>r)$, we simply keep the same interval loop $\mathbf{I}_{r}$, and for $r'\ge r_0$, the data structure is updated to $\mathbf{I}_{r'}$ by certain manipulations with the following three steps:
\begin{enumerate}
	\item[-]
	\underline{Detection} detects the forecast events from the data $\mathbf{I}_{r}$, and updates the event queue;
	\item[-]
	\underline{Processing} deletes and inserts temporarily several intervals related to that event;
	\item[-]
	\underline{Trimming} resolves overlaps of inserted intervals, and generates $\mathbf{I}_{r'}$.
\end{enumerate}
The last step is similar as in the original MMP algorithm, while the first two steps contain several new ideas.
Detection (re-)computes the forecast events and updates the event queue, which involves insertion, deletion and/or replacement of some events. Processing produces a provisional interval loop, which may have overlapped intervals. Trimming makes it a valid interval loop in a true sense.
Our algorithm ends when $W(r)$ collapses to the farthest point or is found to be inconsistent (which occurs at some moment after an occurrence of a bifurcation event or a non-generic event).
Then we obtain the cut-locus $C$ and the distance from $p$ to every vertex passed though.
As for the computational complexity, our algorithm takes $O(n^2 \log n)$ time and $O(n)$ space, where $n$ is the number of vertices of $S$, and it can be modified to be able to find the shortest geodesic with $O(n^2)$ space (see subsection \ref{complexity}).

A particular feature of our algorithm is that, unlike the MMP algorithm, we can visualize the ongoing wavefront propagation, i.e., we compute the set $W(r)$ of points having shortest paths of length $r$ from $p$ {\em all at once}, as well as partially-constructed cut-locus during execution. As a remark, in \cite{Mount} Mount describes how to find the cut-locus $C$ by the information of obtained intervals -- for each face $\sigma$ one can detect $C\cap \sigma$ by computing an associated {\em Voronoi diagram}. However, it requires a bit heavy new task additionally to the MMP and it seems not quite obvious how to implement it to computer program which actually works. In contrast, our algorithm instantly produces the complete information of the cut-locus $C$.

In general, bifurcation events may occur, and then $W(r)$ break off into several connected components.
This phenomenon is the most difficult obstacle for tracing the wavefront propagation beyond the MMP algorithm.
In fact, our algorithm is designed to depend only on the {\em local data} (data of neighboring arcs participating in the wavefront), not {\em global data} of the wavefront, and therefore, our implemented program may stop at a certain moment after some bifurcation event actually happens. In this sense, our algorithm is surely limited. Nevertheless, it seems that there has not been known other practical approach accompanied by actual implementation, as far as the authors know.

As a generalization in different direction, we can use a modified version of intervals to build data structure from a point $p$ on a polyhedral surface and a positive real number $r$, for query of {\em enumeration of all geodesics} shorter than $r$, from $p$ to arbitrarily chosen point $q$ on the surface (this also works for non-convex case as well). This generalization will be dealt with in another paper \cite{Tateiri2}.

\section{Preliminaries}

Throughout the present paper, let $S$ be the boundary surface of a compact 3D convex polyhedral body, i.e., $S$ is a $2$-dimensional polyhedron (= the realization of a finite simplicial complex) embedded in Euclidean space $\R^3$ such that it is homeomorphic to the standard $2$-sphere and that every vertex $v$ of $S$ is {\em elliptic}, i.e.,
the sum of angles around $v$ (measured along the faces) is less than $2\pi$ \cite{Alex,DO}.
We assume that every face of $S$ is an oriented triangle so that the orientation is anti-clockwise when one sees the 3D body from outside.

A {\em path} between $p$ and $q$ (on $S$) is a piecewise linear path connecting these points on the polyhedron $S$.
Among all paths between $p$ and $q$, we can consider a shortest one (there may be multiple shortest paths between $p$ and $q$).
The length of a shortest path between $p$ and $q$ defines the {\em distance} $d(p, q)$ on $S$.
A shortest path $\gamma$ satisfies the following properties:
\begin{itemize}
	\item[-] $\gamma$ is straight on any face which it meets, and when $\gamma$ passes through an interior point of an edge, $\gamma$ will be straight on the unfolding obtained from two faces attaching the edge;
	\item[-] $\gamma$ never passes through any vertex.
\end{itemize}
A {\em geodesic} on a polyhedron $S$ is defined as a (not necessarily shortest) path satisfying the same properties as above.
Obviously, a geodesic on $S$ is a locally shortest path.

Given a convex polyhedron $S$ and a point $p \in S$, the wavefront $W(r)\, (=W(r, p))$ for $r>0$ and the cut-locus $C(S, p)$ are defined as in Introduction.
As $r$ increases, the geometric shape of the wavefront changes.

\begin{definition}\upshape
\label{geom_events}
{\bf (Geometric events)}
We define several {\em events} of the wavefront propagation at $r=r_0$ as follows:
\begin{enumerate}
	\item[(v)] a {\em vertex event} occurs when $W(r_0)$ hits a vertex of $S$;
	\item[(e)] an {\em edge event} occurs when a ridge point of $W(r_0)$ hits an edge;
	\item[(c)] a {\em collision event} occurs when multiple ridge points of $W(r)$ $(r<r_0)$ collide at once, and result in a single ridge point of $W(r')$ $(r_0\le r')$, or (a component of) the wavefront converges at the point and disappears.
	\item[(b)] a {\em bifurcation event} occurs when $W(r_0)$ intersects itself and breaks off into several pieces for $r>r_0$.
\end{enumerate}
\end{definition}

We divide edge events (e) into the following two patterns. Let $A$ and $B$ be neighboring arcs in $W(r_0)$ joined by the ridge point $a$ which hits an edge $e$. Unfold the two faces incident to $e$, and divide the plane by the line containing $e$. We set
\begin{enumerate}
	\item[(ec)] a {\em cross event}: if the centers of $A$ and $B$ are located in the same half-plane;
	\item[(es)] a {\em swap event}: if the centers of $A$ and $B$ are located in opposite half-planes.
\end{enumerate}
Furthermore, among collision events, we distinguish the following special one:
\begin{enumerate}
	\item[(cf)] the {\em final event} occurs when the wavefront reaches the farthest point and disappears.
\end{enumerate}

\begin{remark}\upshape \label{partial_propagation}
At some moment, the wavefront can be tangent to an edge at some point and go through to partially propagate to the next face. We do not include this case into the above list of geometric events, as the arc simply expands on the unfolding along the edge. In other words, our data structure does not need to be changed. As seen later (\S \ref{ispropagated}), it makes our algorithm much simpler, while our instantaneous visualization of the wavefronts does not depict this partial propagation precisely (but it does not affect the calculation of the cut-locus). By this definition, we can ensure that every interval appears exactly \textit{once} in the interval loop, and every interval with non-empty true extent (see \S \ref{basic_data_structure}) is involved in exactly \textit{two} (vertex, edge or swap) events, where it is propagated in the first one and removed in the second one. Otherwise, it requires special treatment of intervals which appear \textit{twice} in the interval loop, which also have one or more descendants. Also, they would be involved in \textit{three} events, where it is propagated in the first one and removed in the second and third ones, whereas some intervals are still involved in only two events. See also Remark \ref{delay_remark}.
\end{remark}

\begin{definition}\upshape
\label{generic}
{\bf (Generic source point)}
\begin{enumerate}[(i)]
\item We say that the source point $p$ is {\em generic} if the following three properties hold:
\begin{enumerate}[(1)]
	\item $p$ is an interior point of a face,
	\item every ridge point of $W(r)$ for any $r>0$ does not pass through vertices and does not move along an edge,
	\item every collision (including the final) event happens in the interior of a face, and the number of ridge points collide at once is three in the final event and two otherwise.
\end{enumerate}

\item When choosing $p$ to be generic, the wavefront propagation admits only geometric events as depicted in Figure \ref{genericevents}; we call them {\em generic geometric events}.
\end{enumerate}
\end{definition}

\begin{figure}[h]
	\centering
	(v) \includegraphics[height=3.25cm, pagebox=cropbox]{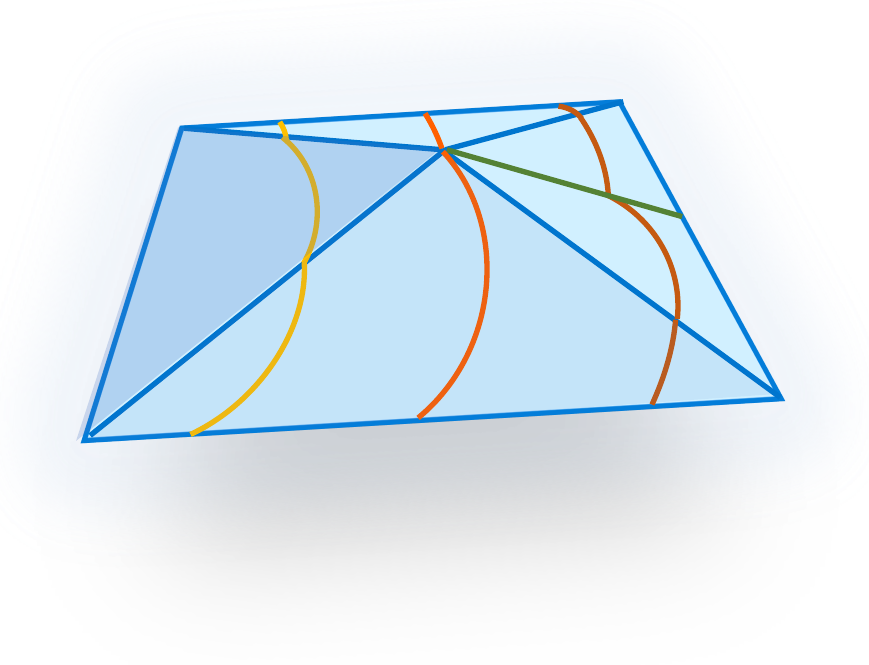}
	(ec)\includegraphics[height=3.25cm, pagebox=cropbox]{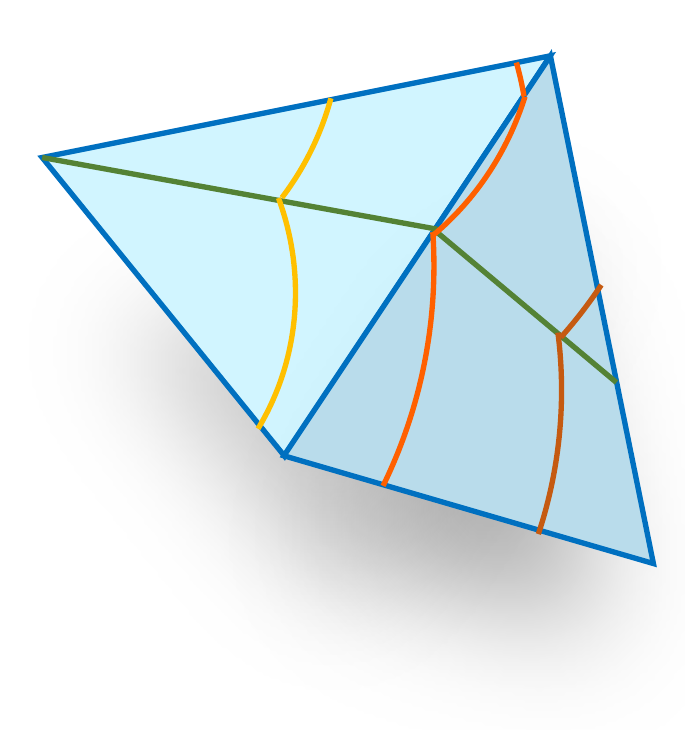}
	(es)\includegraphics[height=3.25cm, pagebox=cropbox]{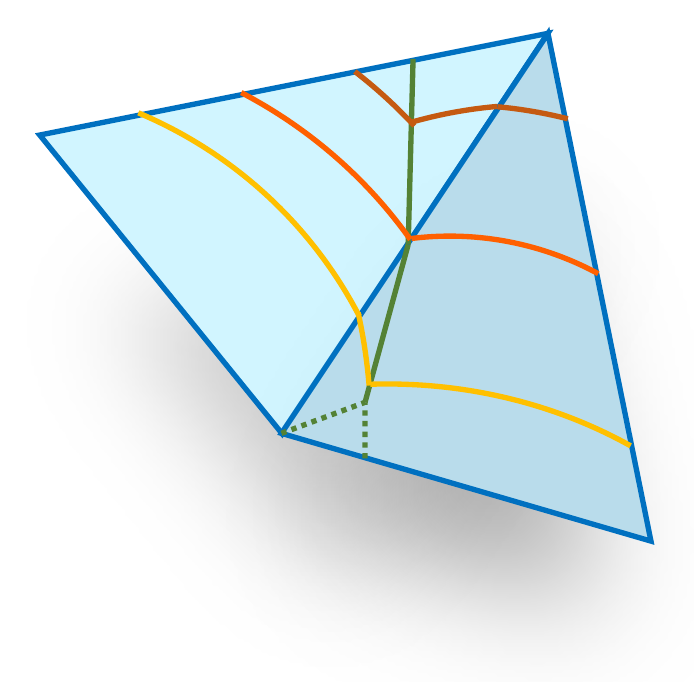} \\
	(c) \includegraphics[height=3.25cm, pagebox=cropbox]{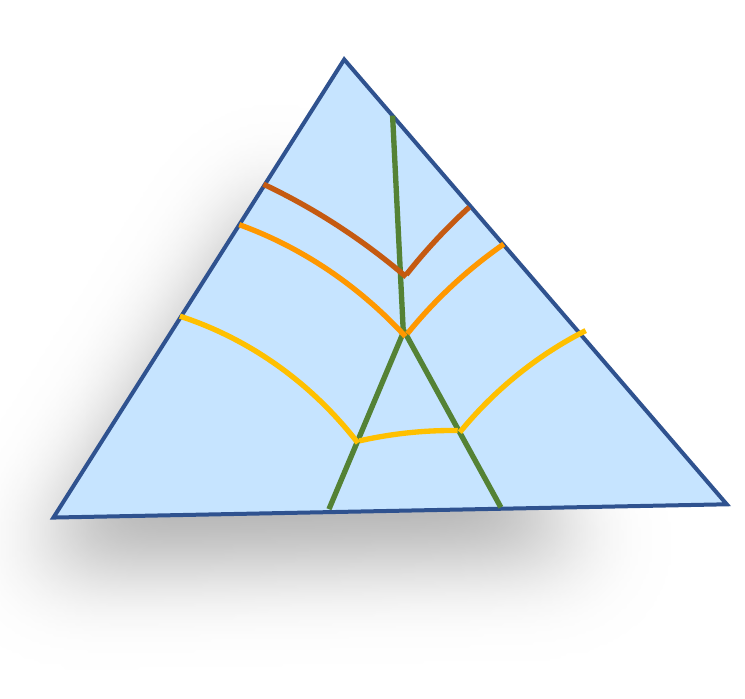}
	(b) \includegraphics[height=3.25cm, pagebox=cropbox]{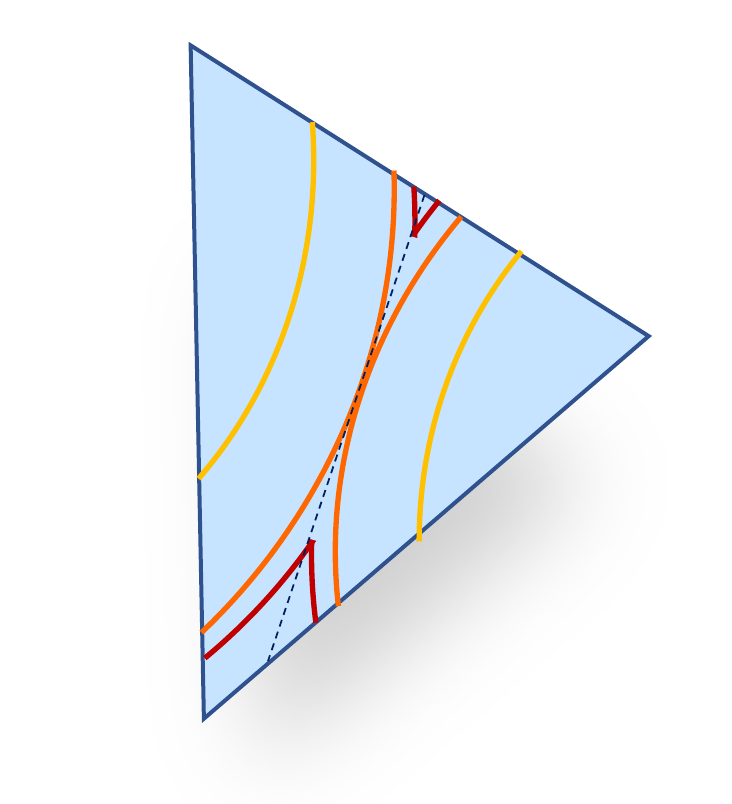}
	(cf)\includegraphics[height=3.25cm, pagebox=cropbox]{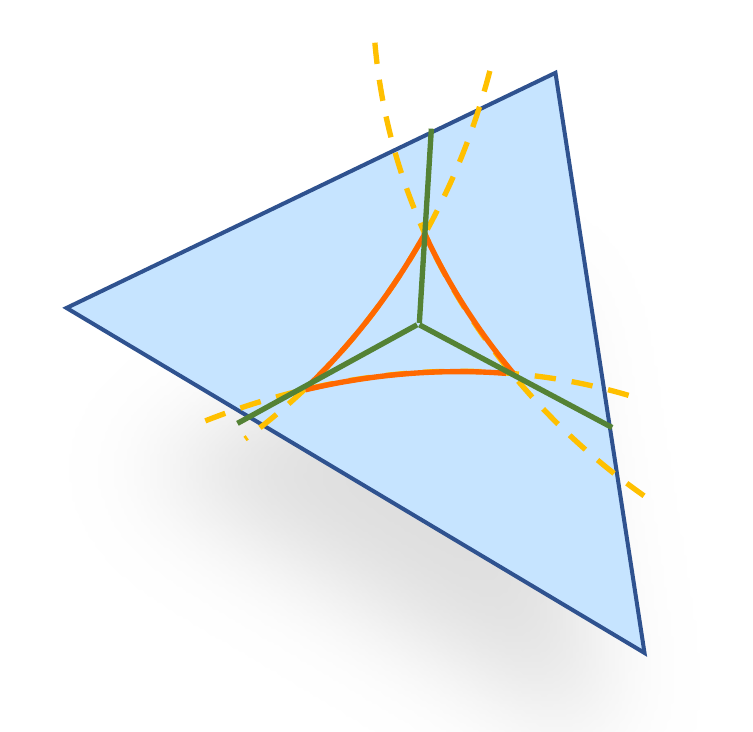}
	\caption{Generic geometric events: the wavefront propagates from the yellow one to the (dark) red one.
	(v) {\em Vertex event}: a new ridge point is created at a vertex, and no ridge point tends to a vertex as $r$ increases; (ec), (es) {\em Cross/Swap edge event}: a ridge point hits an edge (it does not move along the edge); there are two types -- two arcs come across the edge from the same side or from the opposite side; (c) {\em Collision event}: only two ridge points collide at once inside a face; (b) {\em Bifurcation event}: two local components of the wavefront get to be tangent to each other, and breaks off into two pieces; (cf) {\em Final event}: the triangle-shaped wavefront goes to shrink and disappears. }
	\label{genericevents}
\end{figure}

In this paper, as mentioned in Introduction, we consider the wavefront propagation with a generic source $p$ for the period until the final event or a bifurcation event happens.

The cut-locus $C=C(S,p)$ is a graph embedded on $S$ whose edges are linear segments.
If there happens a vertex event at a vertex $v$ with $r=r_0$, the propagation around $v$ creates the cut-locus, i.e., $W(r_0-\varepsilon)$ for small $\varepsilon>0$ on the unfolding around $v$ is locally one circular arc, while $W(r_0+\varepsilon)$ has one ridge point locally (Figure \ref{ridge_point} in Introduction). Then $v$ is an end of the cut-locus $C$. Therefore, we see that

\begin{lemma}\label{lem1}
If the source $p$ is generic and the bifurcation event does not appear during the wavefront propagation, then the obtained cut-locus $C$ is connected and has a tree structure with leaves at vertices of $S$ and nodes with degree 3, which are points at which collision events and the final event happen.
\end{lemma}

\begin{lemma}\label{lem2}
Generic source points form an open and dense subset of $S$; the complement is the union of all edges and finitely many closed piecewise algebraic curves on $S$.
\end{lemma}

Intuitively, these lemmata look almost trivial, and indeed they are checked practically by the fact that our algorithm properly works (Remark \ref{generic2}). A short proof will be given in Appendix A.

\begin{remark}\upshape \label{generic1}
If the source point $p$ is generic, the shape of the cut-locus $C(S, p)$ is stable with respect to small perturbations of $p$.
Namely, for sufficiently near generic $p'$, $C(S, p)$ and $C(S, p')$ are the same graph so that corresponding two edges have almost the same length.
Now suppose that the source point $p$ is not generic. Even though the cut-locus $C(S, p)$ exists but possesses some degenerate vertices or edges. When perturbing $p$ to a generic $p'$, such degenerate points locally break into generic geometric events as indicated in Figure \ref{genericevents}, and the new cut locus $C(S, p')$ should be sufficiently close to $C(S, p)$.
\end{remark}

\begin{remark}\upshape \label{generic2}
Theoretically, it is possible to determine whether a chosen point $p$ is generic or not, if we have an unfolding of the whole of $S$ in advance. In our specifications, however, we are not supposed to have such prior information; rather to say, as mentioned before, we are aiming to produce a nice planar unfolding.
In practice, non-generic geometric events do not occur unless we intentionally set up such input of $S$ and $p$.
\end{remark}

\begin{remark}\upshape \label{generic3}
In the contexts of differential geometry and singularity theory, wavefronts, caustics, cut-loci and ridge points on a smooth surface have been well investigated, see e.g., Arnol'd \cite{Arnold}. Our classification of generic geometric events is motivated as a sort of corresponding discrete analog.
\end{remark}

%%%%%%%%%%%%%%%%%%%%%%%%%%%%%%%%%%%%%%%%%%%%%%%%%%%%
\section{Main algorithm}

\subsection{MMP algorithm}
The MMP algorithm \cite{MMP} (and Mount's earlier algorithm \cite{Mount})
 encodes geodesics as the data structure named by {\em intervals}:
\begin{itemize}
	\item Input: a polyhedron $S$ and a source point $p$ on $S$.
	\item Output: a set of intervals for each edge, which enables us to find the shortest geodesic from $p$ to any given point $q$ on $S$.
	\item Complexity: $O(n^2 \log n)$ time, $O(n^2)$ space, where $n$ is the number of edges of $S$.
\end{itemize}
An interval $I$ is a segment, called the {\em extent} of $I$, in an edge $e$ of $S$ endowed with additional data being necessary to find the shortest path from $p$ to points of the extent.
Intervals are inductively propagated -- each interval generates a new one (its child interval) step-by-step by manipulations called {\em projection} and {\em trimming}.
The algorithm uses a priority queue to manage the order of the propagation of the intervals. The priority of an interval is the {\em shortest distance} between the source point $p$ and its extent, and any smaller value means to be propagated earlier.

\subsection{Our problem}
Our main problem is to reveal some richer structure of geodesics on $S$ by describing the wavefront propagation interactively and accurately, where we deal with not only a single geodesic from a source point $p$ but also all geodesics from $p$ at once.
At the final moment, we obtain the entire cut-locus $C(S, p)$ and the source unfolding, provided that the bifurcation event does not occur in the whole process; otherwise, our algorithm stops at some moment after that bifurcation event.

Our algorithm runs in the following time and space complexity:
\begin{itemize}
	\item Input: a convex polyhedron $S$ and a point $p \in S$.
	\item Output: the cut locus $C(S, p)$.
	\item Complexity: $O(n^2 \log n)$ time, $O(n)$ space.
\end{itemize}
Furthermore, as an option, our algorithm can also support shortest path query using extra space complexity \cite{Tateiri2};
\begin{itemize}
	\item Input: a convex polyhedron $S$ and a point $p \in S$.
	\item Output: the cut locus $C(S, p)$ and a set of intervals for each face, to be able to find shortest geodesic from $p$ to any given point $q$ on $S$.
	\item Complexity: $O(n^2 \log n)$ time, $O(n^2)$ space.
	\item Input of query: a point $q$ on $S$.
	\item Output of query: the shortest path(s) from $p$ to $q$.
\end{itemize}

\subsection{Data structure}
The wavefront $W(r)$ is an oriented closed embedded curve on $S$ consisting of {\em circular arcs} on faces.
For each circular arc $A$,
we introduce the notion of an {\em enriched interval} $I\, (=I_A)$ as a data structure to express the arc $A$ equipped with some additional data.

\begin{definition}\upshape
\label{interval}
We define an {\em enriched interval} $I$ as a data structure shown in Table 1. Each item is denoted by $I$.$[--]$ for notational convention.

\begin{table}[h]
\centering
\begin{tabular}{l | l}
\hline
	$I.\Face$ & the oriented face $\sigma$ which contains the arc $A$ \\
	$I.\Center$ & the center $p_A$ of the arc $A$ on the plane $H_\sigma$ containing $\sigma$ \\
	$I.\Edge$ & the oriented edge $e$ of $\sigma$ into which $A$ is projected from $p_A$ \\
	$I.\Extent$ & the (foreseen) extent $e_A$ associated with $A$\\
	$I.\Prev$ & the enriched interval associated with the previous arc connecting to $A$\\
	$I.\Next$ & the enriched interval associated with the next arc connecting from $A$ \\
	$I.\Ridge$ & the ridge point to which $A$ is adjacent as the start point \\
	$I.\Parent$ & the enriched interval which generates $I$\\
\hline
\end{tabular}
\vspace{12pt}\\
\caption{An enriched interval $I$}
\label{table1}
\end{table}
We also define an {\em interval loop}
$$\mathbf{I}_r=\left\{ I_0^{(r)}, I_1^{(r)}, \cdots, I_{k(r)}^{(r)}\right\}$$
to be a finite sequence of enriched intervals that satisfy
\begin{center}
$I_i^{(r)}.\Prev=I_{i-1}^{(r)}$ and $I_i^{(r)}.\Next=I_{i+1}^{(r)}$
\end{center}
for $0 \le i \le k(r)$, where we put $I_{k(r)+1}^{(r)}:=I_0^{(r)}$ and $I_{-1}^{(r)}:=I_{k(r)}^{(r)}$.
\end{definition}

An enriched interval is similar but different from the notion of an {\em interval} used in Mount's algorithm \cite{Mount} and the MMP \cite{MMP}. Main differences are, e.g.,
\begin{itemize}
\item[-] all enriched interval in the wavefront make up an interval loop;
\item[-] an interval loop is a circular doubly-linked list: each enriched interval has the previous and the next interval corresponding to adjacency of arcs and orientation of the wavefront;
\item[-] our enriched interval depends on $r$;
\item[-] an enriched interval may have the {\em empty} extent with non-trivial additional data.
\end{itemize}

Each item in Table \ref{table1} in Definition \ref{interval} depends on $r$;
those are created at the time when the corresponding arc $A$ is born,
and are valid until $A$ disappears.
In particular,
\begin{itemize}
\item[-] the data in $I.\Face$, $I.\Center$, $I.Edge$ and $I.\Parent$ are fixed when $A$ is born;
\item[-] the data in $I.\Extent$, $I.\Prev$, $I.\Next$ and $I.\Ridge$ are updated at every moment where some geometric event involving $A$ happens.
\end{itemize}
Below we explain the meaning of each item in Table \ref{table1}.

\subsubsection{An arc}
To begin, let $r_1>0$ be fixed.
Suppose that a circular arc $A$ participating in $W(r_1)$ lies on a face (oriented triangle) $\sigma$ of $S$.
Let $H_\sigma$ denote the affine plane containing $\sigma$ in $\R^3$, then there is a unique point $p_A \in H_\sigma$ such that $A$ is an arc in the circle on $H_\sigma$ centered at $p_A$ with radius $r_1$.
Take a point $q \in A$ and the shortest path $\gamma$ on $S$ from $p$ to $q$.
We find an unfolding of $S$ expanded on $H_\sigma$, on which $\gamma$ is represented by a line segment, as shown in Figure \ref{net}. The 3D coordinates of the point $p_A$ is explicitly obtained from the 3D coordinates of $p \in S$ by inductively operating certain rotations of $\R^3$ along edges which $\gamma$ intersects.

\begin{figure}[h]
	\centering
	\includegraphics[height=4cm, pagebox=cropbox]{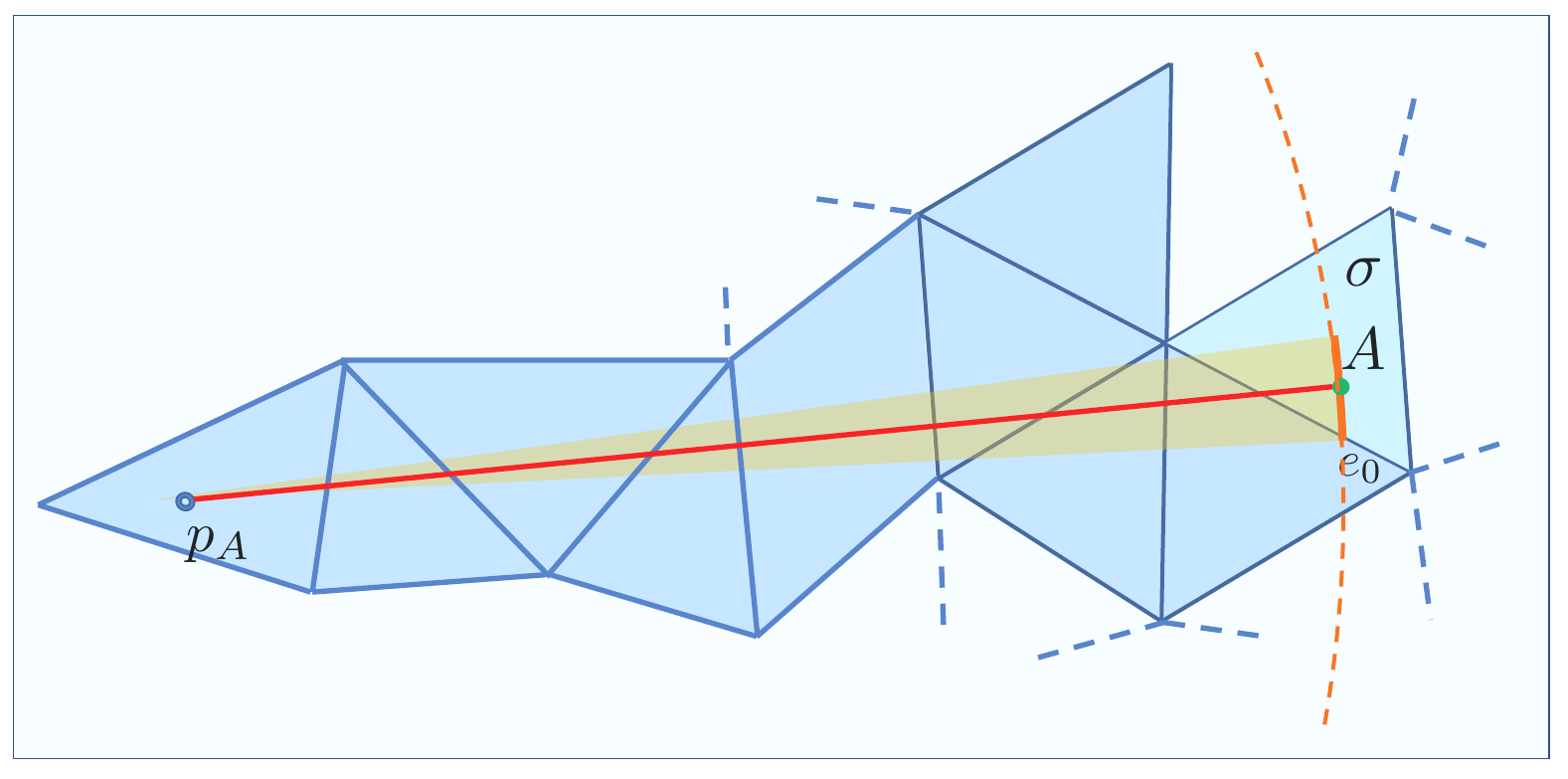}
	\includegraphics[height=4cm, pagebox=cropbox]{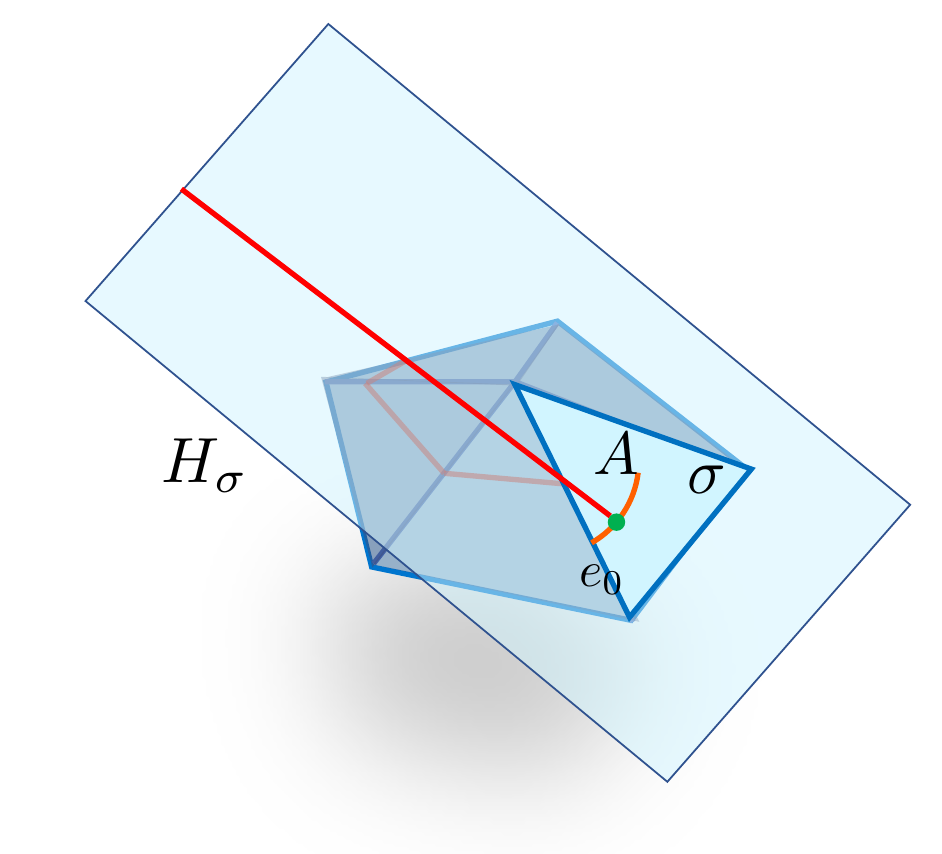}
	\caption{A circular arc on a face $\sigma$ (the left depicts an unfolding on the plane $H_\sigma$). }
	\label{net}
\end{figure}

\subsubsection{Basic data structure for an arc}
Suppose that $p_A$ is outside $\sigma$, and only one edge of $\sigma$, say $e_0$, cuts any segments between $p_A$ and points of $A$. The rays from $p_A$ to the arc $A$ meet another edge of $\sigma$ in opposite side of $e_0$ with respect to the location of $A$.\label{basic_data_structure}
\begin{enumerate}
\item Suppose that $I = I_A$ is projected from the center $p_A$ to only one edge $e$ and yields $I_1$ (see the left picture of Figure \ref{projection}).
The direction of $e$ is chosen to be compatible with the orientation of $A$. We first define the {\em true extent} associated with $A$ by the subrange $e_A \subset e$ which $A$ will actually pass through, see Figure \ref{extent}.
It can be the empty set (see the right picture of Figure \ref{extent}).
Notice that the true extent $e_A$ is fixed after all the events involving $A$ have occurred.
Therefore, in the middle of the process, what we can do is only to provisionally find a {\em foreseen extent}, which we denote by $\tilde{e}_A$, and update it just after the next event happens (indeed, this procedure is the heart of our algorithm, which will be described in detail later in the following sections).
For now, we put
$$I=I_A:=(\sigma, e, p_A, \tilde{e}_A),$$
where each is referred to as
$$I.\Face = \sigma, \;\; I.\Edge=e, \;\; I.\Center = p_A, \;\; I.\Extent = \tilde{e}_A$$
and we will append some additional data to this data structure and use the same notation $I$ or $I_A$; we call it an {\em enriched interval} or simply {\em interval}.
\item
Suppose that $I = I_A$ is projected from $p_A$ to two edges $e_1$ and $e_2$
whose order is determined by the orientation of $\sigma$ and yields $I_1$ and $I_2$ (see the right picture of Figure \ref{projection}).
Then $A$ is divided into two pieces, say $A_1, A_2$, projected into $e_1, e_2$, respectively.
If each sub-arc $A_i$ has the non-empty foreseen extent, $\tilde{e}_{A_i}\not=\emptyset \subset e_i$ ($i=1,2$),
then we associate to $A$ an ordered pair of two enriched intervals
$$I_1-I_2 \quad \mbox{with} \;\; I_i=I_{A_i} :=(\sigma, e_i, p_A, \tilde{e}_{A_i}) \;\;\;\; (i=1,2).$$
We say that $I_1$ and $I_2$ are twins.
\item
Suppose that the source point $p$ is an interior point of $\sigma$ with edges $e_0, e_1, e_2$ (anti-clockwise).
We then define the {\em initial loop} by a triple of enriched intervals
$$I_0-I_1-I_2-I_0 \quad \mbox{with} \;\; I_i :=(\sigma, e_i, p, e_i) \;\;\;\; (i=0,1,2).$$
\end{enumerate}

\begin{figure}
	\centering
	\includegraphics[height=4cm, pagebox=cropbox]{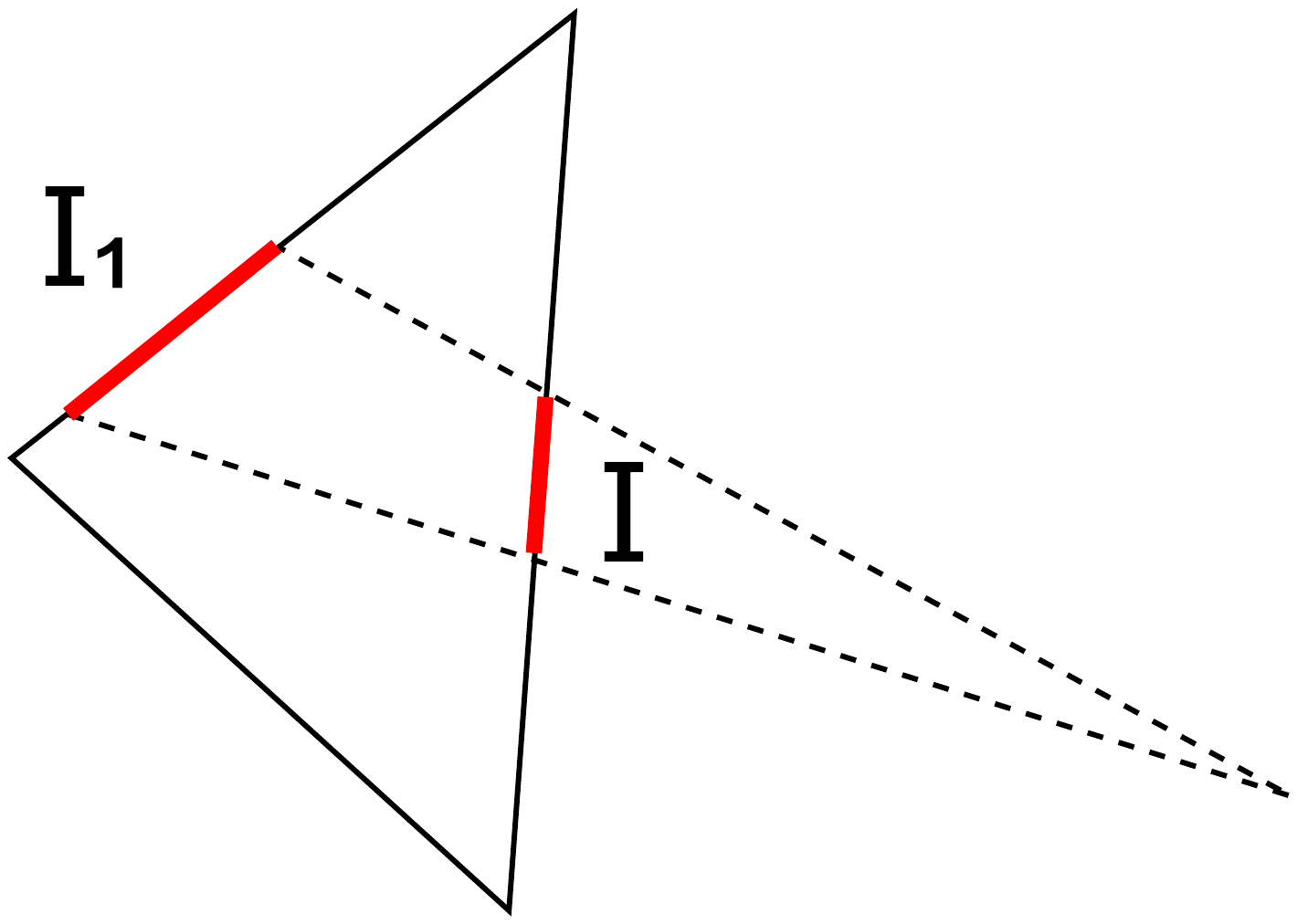}
	\includegraphics[height=4cm, pagebox=cropbox]{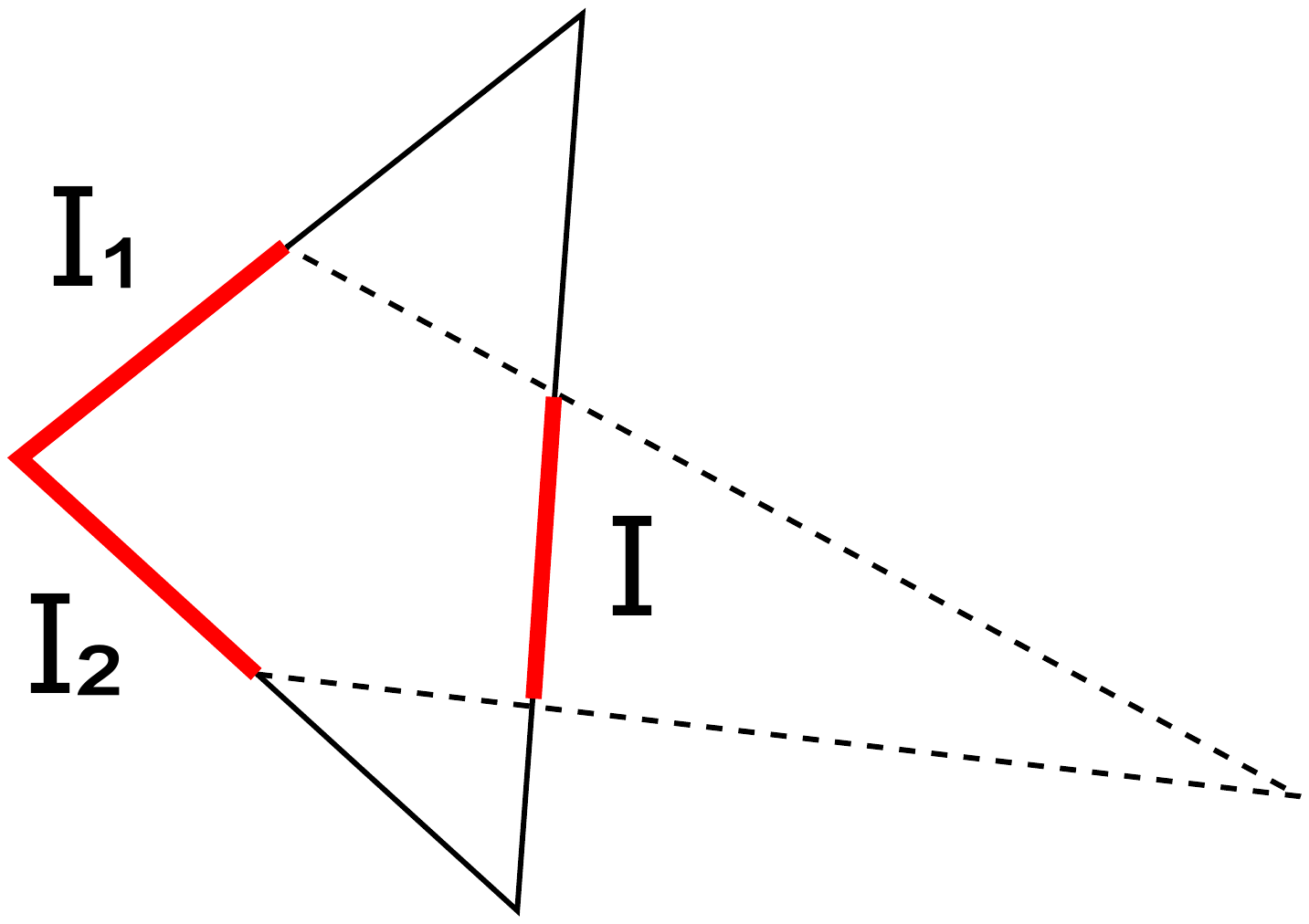}
	\caption{Interval $I$ is projected from $I.\Center$ to opposite edges. }
	\label{projection}
\end{figure}

\begin{figure}
	\centering
	\includegraphics[height=3.75cm, pagebox=cropbox]{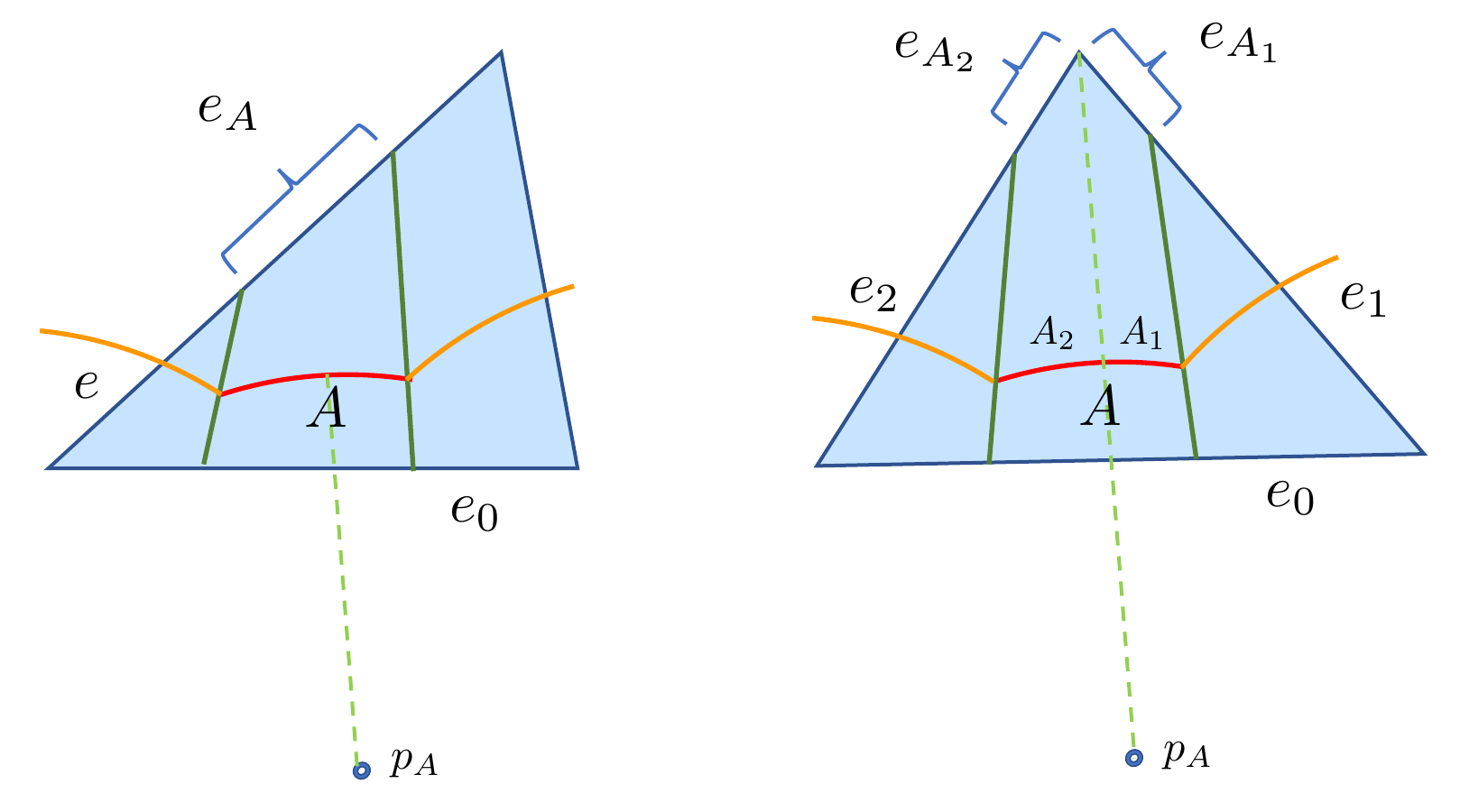}\qquad
	\includegraphics[height=3.75cm, pagebox=cropbox]{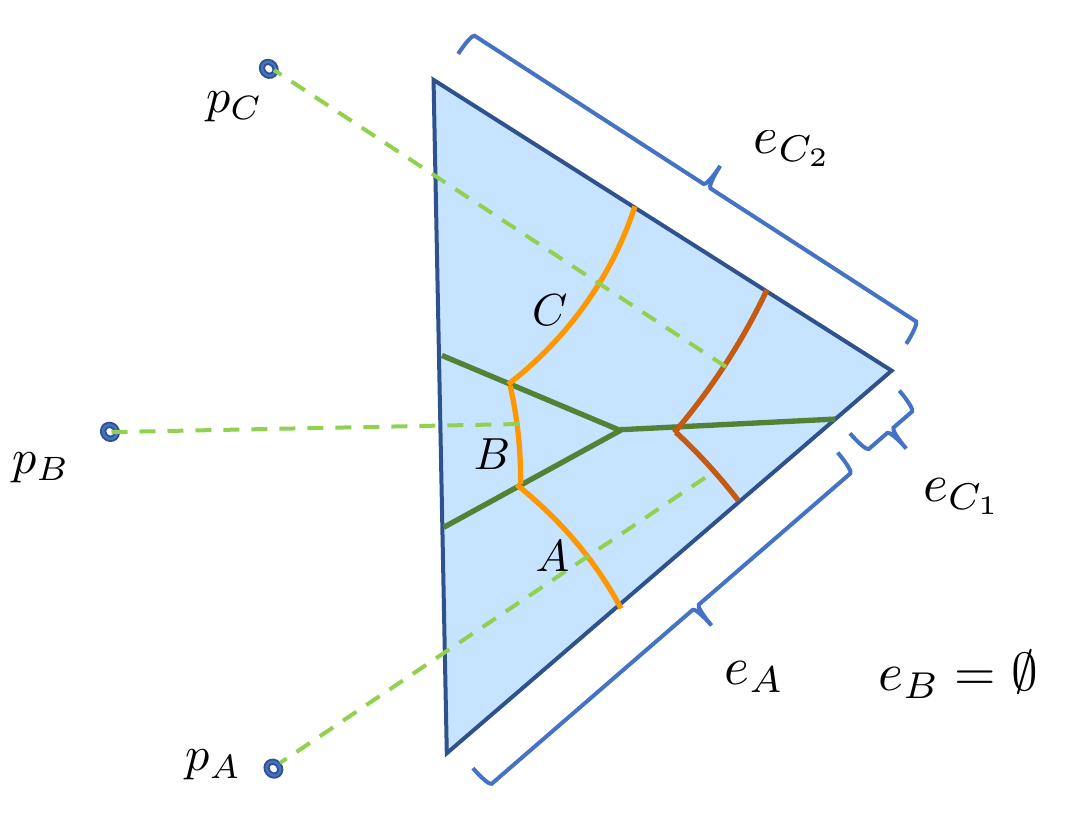}
	\caption{True extents}
	\label{extent}
\end{figure}

\subsubsection{Parent of an arc} \label{parent}
Let $I_A$ be an enriched interval associated with an arc $A$ (or one of twins) in a face $\sigma$.
If the true extent is non-empty, then the arc $A$ will pass through the extent and go into the next face $\tau$ sharing the edge with $\sigma$.
Let $A'$ be the new resulting arc in $\tau$, and $I_{A'}$ the corresponding interval for $A'$
(or $I_{A'_1}-I_{A'_2}$ if $A'$ has two associated enriched intervals).
We say that $I_A$ is propagated to $I_{A'}$, and also call $I_A$ its {\em parent}:
$$I_{A'}.\Parent:=I_A, \qquad (\mbox{or}\;\; I_{A'_i}.\Parent:=I_A\;\; (i=1,2)).$$
Throughout this paper, $A', B', B'', \cdots$ mean the resulting arcs to which their parents $A, B, B', \cdots$, respectively, are propagated.

This item is used as follows. For instance, if an enriched interval $I$ satisfies
$$I.\Prev = I.\Parent,$$
then we understand that the parent is non-empty, say $A$, and
the arc represented by $I$ is nothing but the resulting arc $A'$ to which $A$ is propagated.
Here, $I.\Ridge$ must be empty, for the endpoint of $A'$ is not a ridge point.
If neighboring intervals $I-J$ have the same parents
$$I.\Parent=J.\Parent,$$
then they are twins (i.e., $I=I_1$ and $J=I_2$).

\subsubsection{Previous/next arcs and ridge points}
An arc $A$ connects with two other arcs in the same wavefront; according to the orientation of the wavefront,
let $B$ be the previous arc and $C$ the next one. Then, for their intervals, we write
$-I_B-I_A-I_C-$ and set
$$I_A.\Prev:=I_B, \;\; I_A.\Next:=I_C.$$
Here $B, A$ or $A, C$ can be a twin.
Let $a$ and $c$ be the joint point $B \cap A$ and $A \cap C$, respectively.
We make a convention that any information of $a$ (resp. $c$) will be appended to and stored in $I_A$ (resp. $I_{C}$).
For instance, if $a$ is a ridge point, we set
$$I_A.\Ridge = a.$$
If not, this item is \textit{Nil}.
When some events happen, items $I.\Prev$, $I.\Next$ and $I.\Ridge$ may be updated.

We recall that there are two patterns of edge events as noted in \S 2; let $I_A-I_B$ be neighboring intervals where arcs $A$ and $B$ are joined by a ridge point $a=I_B.\Ridge$ which hits an edge. Then the two patterns can easily be distinguished as follows:
\begin{itemize}
	\item \underline{\em Cross event}: both extents of $I_A$ and $I_B$ are non-empty;
	\item \underline{\em Swap event}: one of the extents of $I_A$ or $I_B$ is empty.
\end{itemize}
Furthermore, we call a swap event to be of type CW (resp. CCW) if $I_A.\Extent$ (resp. $I_B.\Extent$) is empty (clockwise/counter-clockwise). Note that by definition, at least one extent is non-empty.

\subsection{Our algorithm}
A basic idea is to express the wavefront propagation $W(r)$ by updating interval loops step-by-step
$$
\mathbf{I}_{r_0} \Longrightarrow \mathbf{I}_{r_1} \Longrightarrow \mathbf{I}_{r_2} \Longrightarrow \cdots \;\;\;
(0=r_0<r_1<r_2<\cdots).
$$

\begin{enumerate}
\item {\bf Initial loop}:
Suppose that $p$ is in the interior of $\sigma$. Then $\mathbf{I}_0$ is given by $I_0 - I_1 - I_2 - I_0$ of directed edges of $\sigma$.
\item {\bf Events}:
Generic geometric events in \S 2 are interpreted as change of the data structure of enriched intervals, named {\em events}.

In our algorithm, we append two more items to the data structure of each $I$ (Table \ref{table2}).
Since an interval $I$ has at most one associated forecast event, we store it as $I.\Forecast$ if it exists.
Also we use another new item $I.\IsPropagated$ to ask whether the arc has been propagated or not (see \S \ref{processing}).
The default of $I.\Forecast$ is \textit{Nil} (indicating it does not exist), and that of $I.\IsPropagated$ is \textit{No}.

\begin{table}[h]
\centering
\begin{tabular}{l | l}
\hline
	$I.\Forecast$ & the forecast event for $I$ if exists; otherwise, \textit{Nil}\\
	$I.\IsPropagated$ & Yes/No for the inquiry whether or not $I$ has been propagated\\
\hline
\end{tabular}
\vspace{12pt}\\
\caption{Additional items in an interval $I$.}
\label{table2}
\end{table}

\item {\bf Manipulation}:
Each $\mathbf{I}_{r_j}$ is updated to $\mathbf{I}_{r_{j+1}}$ at an event -- some intervals in $\mathbf{I}_{r_j}$ are removed, some new intervals are inserted, and items of remaining intervals are updated.
Each event activates three editing processes named by
{\em Detection}, {\em Processing} and {\em Trimming}.
The algorithm halts when the final event occurs and the wavefront disappears, or when some exception arises, e.g., a bifurcation event or non-generic event happens.
\end{enumerate}

Pseudocode of our algorithm is described in Algorithm \ref{algo1}. Each process will be described below.

{\footnotesize
\begin{algorithm}[h]
\caption{(Main algorithm)} \label{algo1}
\begin{algorithmic}
	\State Create the initial interval loop
	\While {the wavefront exists}
		\For {each interval $I$ whose adjacency changed in the previous step}
			\State Trim temporary extents of $I$ and remove redundant intervals
			\State Detect the forecast and update the event queue
		\EndFor
		\State Process the top-most event in the event queue
	\EndWhile
\end{algorithmic}
\end{algorithm}
}

\subsection{Event Detection} \label{detection}
Suppose that we have an interval loop $\mathbf{I}_{r_j}$ which represents the wavefront $W(r_j)$.
Let $I=I_A$, $I_{A_1}$ or $I_{A_2}$ be an enriched interval belonging to it, associated with an arc $A$ or twin sub-arcs on a face $\sigma=I_A.\Face$.

\subsubsection{Detecting the earliest event for an arc}
We first detect which geometric event for $A$ will happen {\em without any consideration on the propagation of other arcs}.

\begin{itemize}
	\item \underline{\em Vertex event} \\
	By the existence of a pair of twins $I_{A_1}$ and $I_{A_2}$, we foresee that a vertex event will happen at the vertex they meet. We store the event in $I_{A_2}$.
	\item \underline{\em Collision event} \\
	In case that consecutive $I_A$, $I_B$ and $I_C$ have the same face and $I_B.\Extent$ is empty, we foresee that the arc $B$ will collapse and a collision event will happen at the equidistant point (circumcenter) from three points, $I_A.\Center$, $I_B.\Center$ and $I_C.\Center$. We store the event in $I_B$.
	\item \underline{\em Cross event} \\
	In case that consecutive $I_A$ and $I_B$ have the same edge and both extents are not empty, we foresee that the ridge point between them hits the edge and a cross event will happen at the point $I_A.\Extent$ and $I_B.\Extent$ meet. We store the event in $I_B$.
	\item \underline{\em Swap event} \\
	In case that an interval $I$ has the empty extent and the next or previous interval is its parent, we foresee a swap event will happen. There are two types of swap events, CW and CCW, depending on its parent is the previous or next of $I$. We store the event in $I$.
\end{itemize}
For each event, the predicted time can exactly be calculated from the point at which the event will happen.

\subsubsection{Priority queue} \label{queue}
Since the earliest forecast event cannot be modified by other events and occurs next for certain, we use a priority queue to schedule all forecast events by their predicted times and choose the earliest one to be processed; we call it the {\em event queue}.

A pseudocode for Detection of the forecast for each enriched interval of the loop is described in Algorithm \ref{algo2}.
Referring to this queue, we can find the earliest event among the forecast events of all enriched intervals belonging to $\mathbf{I}_{r_j}$.

{\footnotesize
\begin{algorithm}[h]
\caption{(Detection at each enriched interval $I$)}\label{algo2}
\begin{algorithmic}
	\If {$I.\Extent = \emptyset$}
		\If { $I.\Prev.\Face = I.\Face = I.\Next.\Face$}
			\State {A collision event will occur; calculate the predicted time}
		\EndIf
		\If {$I.\Prev = I.\Parent$ and $I.\Next.\Extent\not=\emptyset$}
			\State {A CW swap event will occur; calculate the predicted time}
		\EndIf
		\If {$I.\Next=I.\Parent$ and $I.\Prev.\Extent\not=\emptyset$}
			\State {A CCW swap event will occur; calculate the predicted time}
		\EndIf
	\Else {\quad (i.e. $I.\Extent \not= \emptyset$)}
		\If {$I$ and $I.\Prev$ are twins}
		\State {A vertex event will occur; calculate the predicted time}
		\EndIf
		\If {$I.\Edge=I.\Prev.\Edge$}
		\State {A cross event will occur; calculate the predicted time}
		\EndIf
	\EndIf
\end{algorithmic}
\end{algorithm}
}

\subsection{Event Processing} \label{processing}
As the result of Detection process, now we have the earliest forecast event of the loop $\mathbf{I}_{r_j}$.
First we delete several intervals of $\mathbf{I}_{r_j}$ involved in that event, and then create and insert new intervals,
and make some changes of items in remaining intervals of $\mathbf{I}_{r_j}$.
Below we describe Processing for each type of events in typical situations (in fact, it is often to need to consider several divided cases, but we avoid a messy description here).

\subsubsection{Recognition of propagated arcs}\label{ispropagated}
If an arc $A$ has a non-empty true extent, $e_A\not=\emptyset$, then let $\xi_A$ be the one closer to $A$ of the two endpoints of $e_A$ (if both endpoints have equal distance, take one of them).
Afterwards, just when $A$ arrives at the point $\xi_A$, some event (vertex, cross/swap) happens at that point and $A$ starts to propagate to the next face. At this moment, we update the item $I_A.\IsPropagated$ to be Yes (from No, the default) and create a new interval $I_{A'}$ representing the resulting arc $A'$ and insert it to the interval loop. Here, of course, it can happen that $A$ starts to propagate to multiple faces, and it creates multiple new intervals, say, $I_{A'}, I_{A''}, \cdots$. Afterwards, the arc $A$ arrives at the other endpoint of $e_A$ and some other event happens at that point. At this moment we recognize that all the propagation of $A$ has been done; namely, we remove $I_A$ from the interval loop.
We remark again that we do not take any attention to a `partial propagation' caused by tangency of $A$ with some edge (Remark \ref{partial_propagation}). That makes the algorithm much simpler.

\subsubsection{Temporary extents}\label{temp_exptent}
In Processing at an event, each newly-created interval $I_A$ may be assigned incomplete data at some items.
For instance, if an arc $A$ is resulted by propagating an arc (= the parent of $A$),
the data structure $I_A$ is created at that moment and the item $I_A.\Extent$ is temporarily filled in by the image of its parent's extent $I_A.\Parent.\Extent$ via the projection from the center $p_A=I_A.\Center$ (Figure \ref{temporary_extents}).
Such a temporary extent for $A$ may have overlaps with extents of previous/next intervals.
The next process Trimming corrects the overlaps and produces a foreseen extent $\tilde{e}_A$ (\S \ref{trim}).
Afterwards, at every event which is related to $A$, $I_A.\Extent$ is updated by new $\tilde{e}_A$,
and finally, if no more related event occurs, then the latest $\tilde{e}_A$ means the true extent $e_A$. In Figure \ref{temporary_extents_b}, we explain a consecutive process updating the extents; the detail of manipulation at each event will be described below.

\begin{figure}
	\centering
	\includegraphics[height=5cm, pagebox=cropbox]{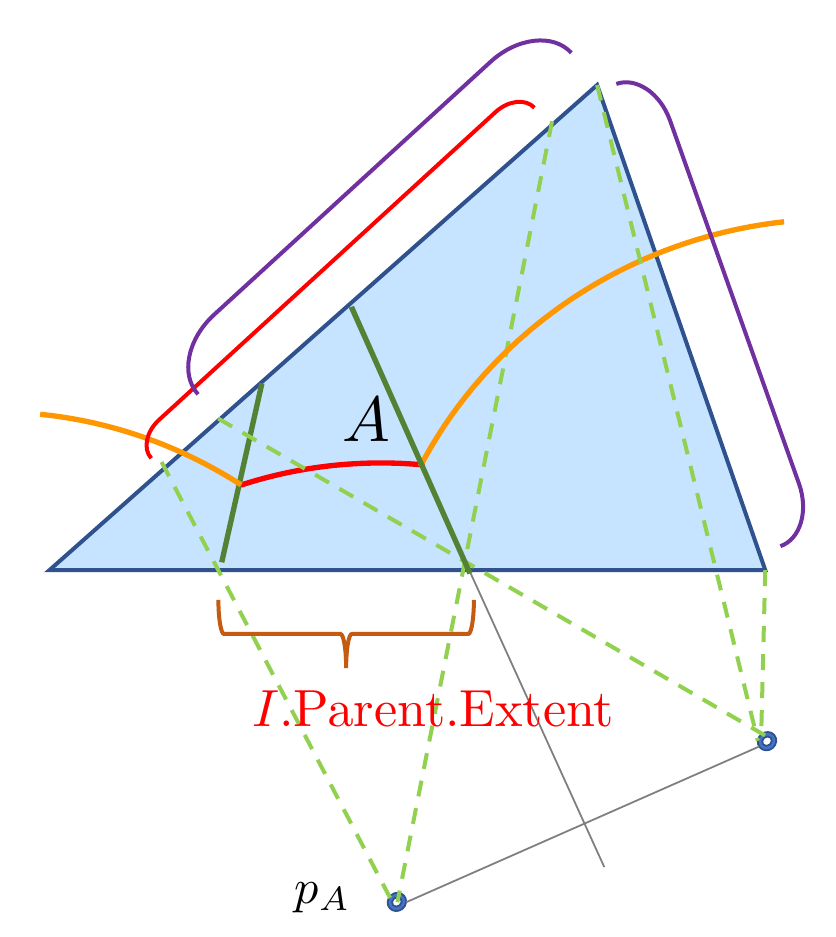}
	\caption{Temporary extents in Processing; overlaps are resolved in Trimming}
	\label{temporary_extents}
\end{figure}

\begin{figure}
	\centering
	\includegraphics[height=4.75cm, pagebox=cropbox]{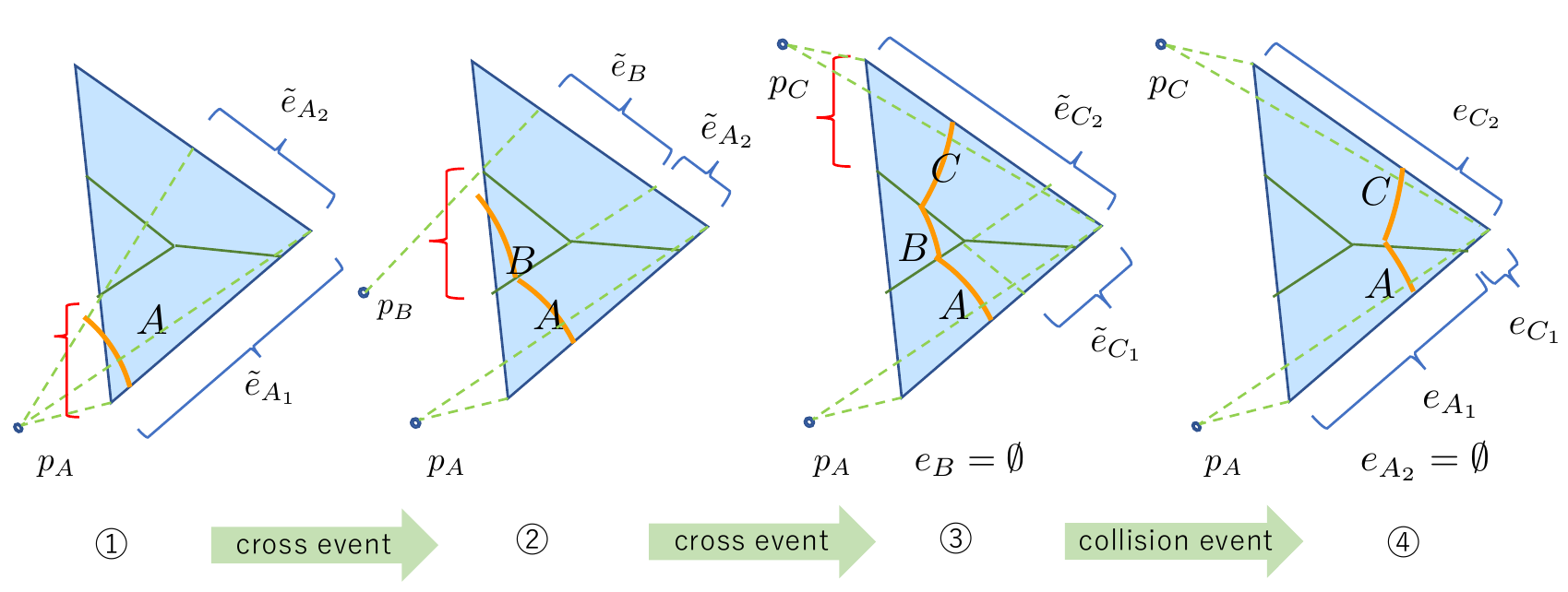}
	\caption{An example of updating $I.\Extent$ at consecutive several events.
	(1) An arc $A$ is born on this face (propagated from its parent) and temporary extents for a twin is made; those do not need to be edited yet.
	(2) A cross event has just created a new arc $B$; Processing puts a temporary extent in $I_{B}.\Extent$, but soon after, Trimming updates $I_{A_2}.\Extent$ and $I_{B}.\Extent$ by resolving overlaps.
	(3) Another cross event has created the third arc $C$; after Trimming, it turns out that $I_{B}.\Extent=\emptyset$. Since $A$ and $C$ are not neighboring, foreseen extents of $A_i$ and $C_i$ ($i=1,2$) may have overlaps, and do not need to be edited yet.
	(4) The next is a collision event; Processing deletes $I_B$ and Trimming edits all extents and deletes $I_{A_2}$ (resolve the redundant twin). 	There remains $-I_{A_1}-I_{C_1}-I_{C_2}-$ in the interval loop. Finally, $I_{A_1}$ will soon be deleted at the coming cross event.
	}
	\label{temporary_extents_b}
\end{figure}

\subsubsection{Cross event}\label{cross}
Consider the cross event such that both neighboring arcs joined by the ridge point, say $A$ and $B$ in order, have {\em non-empty true extents} just before the event happens. The edge (blue) is locally divided into the two extents. Each of arcs $A$ and $B$ has two patterns according to whether it has already been propagated or not yet; this information (Yes or No) is stored in $I.\IsPropagated$. There are four patterns, see Figure \ref{cross_event}.
In the first and second ones, we simply remove the interval $I_B$ (resp. $I_A$) from the interval loop, and instead, suitably insert a new interval $I_{A'}$ (resp. $I_{B'}$) to produce the new loop, e.g., in case of (No-Yes), we do the manipulation
$$-I_{A}-I_{B}-I_{B'}-
\quad \Longrightarrow \quad -I_{A}-\dbox{$I_{A'}$}-I_{B'}-$$
For the third pattern, just before the cross event happens, both arcs $A$ and $B$ have not yet arrived at $\xi_A$ and $\xi_B$, respectively, thus both $I.\IsPropagated$ are still being `No'. If both arcs have already been propagated, that is the fourth one.

\begin{figure}[h]
	\centering
	\includegraphics[height=3.25cm, pagebox=cropbox]{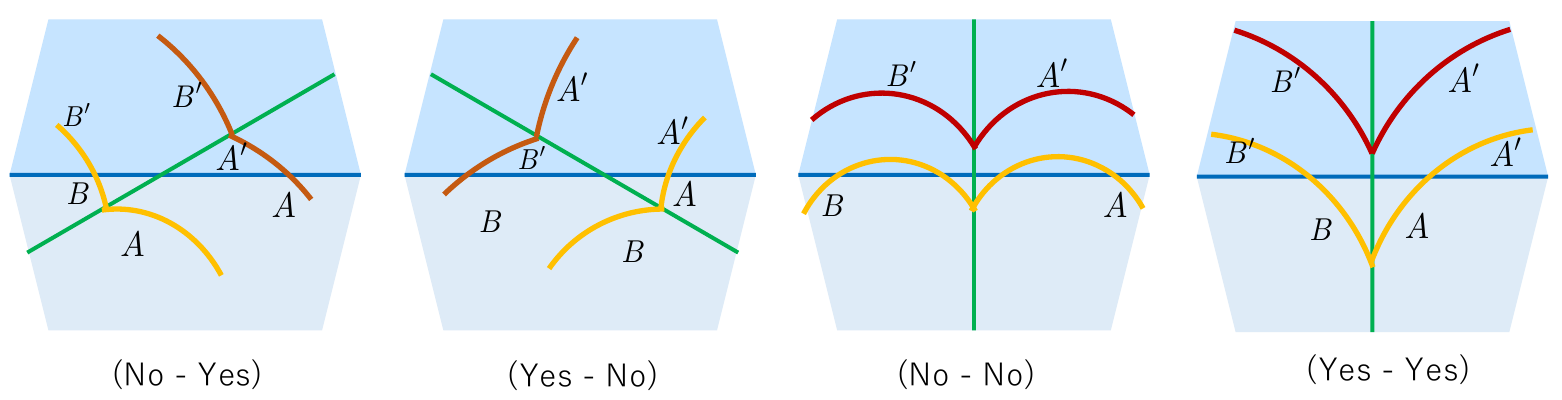}
	\caption{There are four patterns for cross events where the cut-locus (green) intersects with an edge (blue),
	referring to ($I_A.\IsPropagated$ - $I_B.\IsPropagated$). }
	\label{cross_event}
\end{figure}

\subsubsection{Swap event} \label{swap_e}
At a swap event, let $A, B$ be the neighboring arcs joined by a ridge point which hits an edge (Figure \ref{swap_event}). Then, one of them has already been propagated, say it $A$; just before the event, $A'$ is joined with $B$ and the extent of $A'$ is empty.
Just after the event, $A'$ disappears and an arc $B'$ newly arises. Namely, arcs $A'$ and $B'$ are swapped. The manipulation on enriched intervals is as follows:
$$-I_{A}-I_{A'}-I_{B}-
\quad \Longrightarrow \quad -I_{A}-\dbox{$I_{B'}$}-I_{B}-\quad (CW)\;\; $$
$$-I_{A}-I_{B'}-I_{B}-
\quad \Longrightarrow \quad -I_{A}-\dbox{$I_{A'}$}-I_{B}-\quad (CCW)$$

\begin{figure}
	\centering
	\includegraphics[height=4.5cm, pagebox=cropbox]{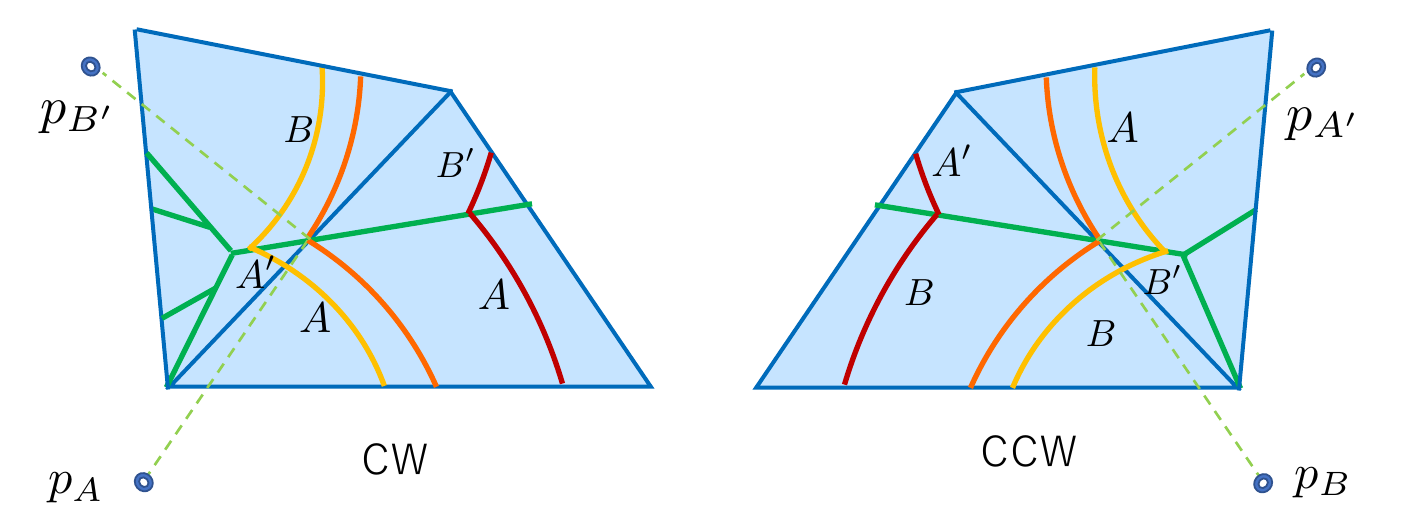}
	\caption{CW/CCW swap events}
	\label{swap_event}
\end{figure}

\subsubsection{Vertex event}
Suppose that an arc $A$ in $\sigma$ meets a vertex $v$ of $\sigma$
and the cut-locus is created in another face $\tau$.
The arc $A$ is represented by twin arcs, say $A_1, A_2$, whose extents are joined at $v$.
For example, look at the left picture of Figure \ref{vc} which is an unfolding around $v$ on the plane $H_\tau$.
Just before the event happens, the twins have already been propagated, and moreover $A_2'$ has also been propagated,
so intervals $I_{A_1'}$, $I_{A_2'}$ and $I_{A_2''}$ exist. Then we do the manipulation
$$-I_{A_1'}-I_{A_1}-I_{A_2}-I_{A_2'}-I_{A_2''}-
\quad \Longrightarrow \quad
-I_{A_1'}-\dbox{$I_{A_1''}-I_{A_2'''}$}-I_{A_2''}-$$

\begin{figure}
	\centering
	\includegraphics[height=5cm, pagebox=cropbox]{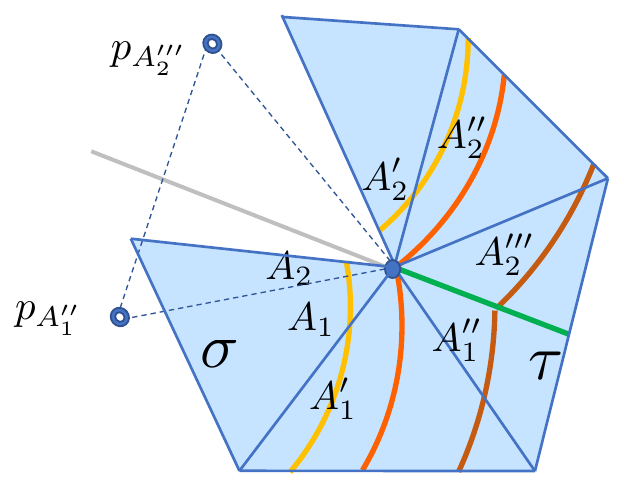} \qquad
	\includegraphics[height=5cm, pagebox=cropbox]{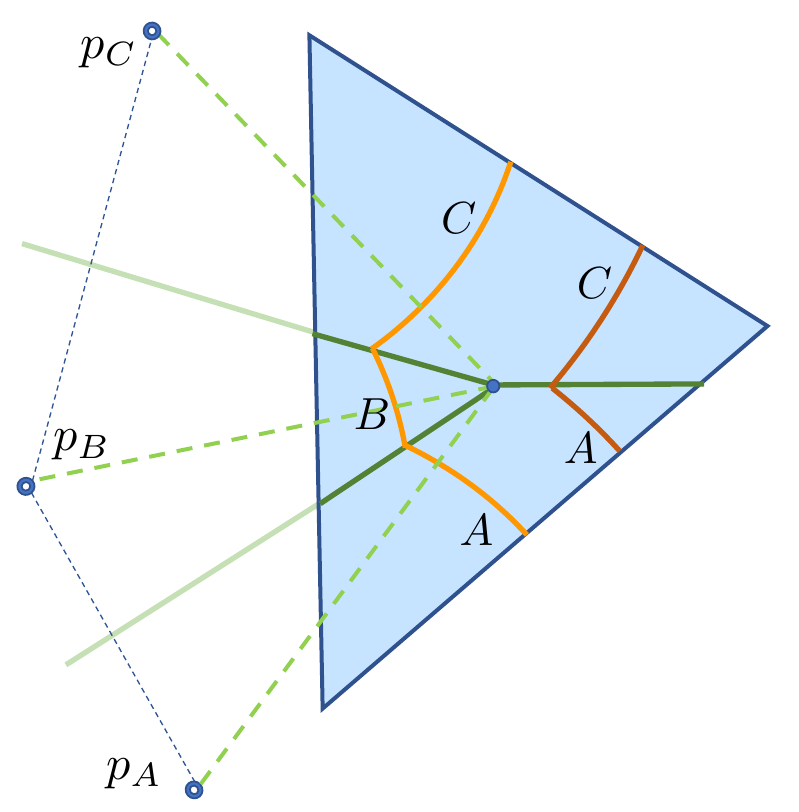}
	\caption{Vertex event and collision event}
	\label{vc}
\end{figure}

\begin{remark}\upshape
Our algorithm may belatedly detect a vertex event which has actually occurred in the past.
This delay is due to our simplification rule to ignore the tangency of the wavefront and an edge.
In Figure \ref{delay}, the arc $A$ gets to be tangent to the edge $e$, and soon after, it meets a vertex $v$,
but we do not recognize this `partial propagation' of $A$, because $I_A.\IsPropagated$ is still being `No'.
When $A$ reaches the endpoint $\xi_A$ of its extent $e_A\subset e$, we update $I_A.\IsPropagated$ to be `Yes', and only then  new twins $I_{A_1}-I_{A_2}$ are recognized. The next Detection step now detects this vertex event at $v$.
Processing and Trimming perform and draw the cut-locus created at $v$ belatedly.
\label{delay_remark}
\end{remark}

\begin{figure}
	\centering
	\includegraphics[height=4cm, pagebox=cropbox]{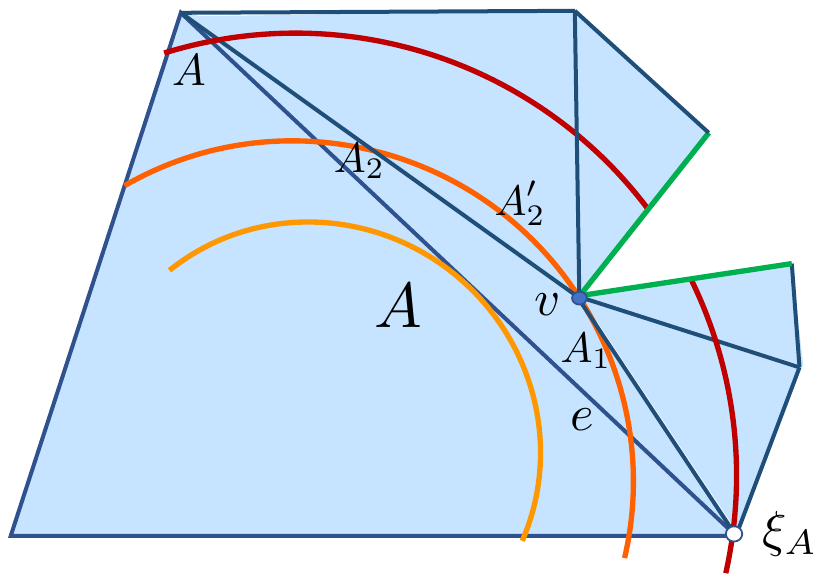}
	\caption{Delayed vertex event}
	\label{delay}
\end{figure}

\subsubsection{Collision event}
A collision event happens when three consecutive arcs, say $A, B, C$ in order, lie on the same face and $B.\Extent$ is empty (Figure \ref{vc}, the right).
Note that the final event is detected as three collision events that occur at the same point. If $W(r)$ consists of only three arcs and intervals, the collision event is the final event: just stop the algorithm.
Otherwise, the manipulation is simply to delete $I_B$:
$$-I_{A}-I_{B}-I_{C}-
\quad \Longrightarrow \quad -I_{A}-I_{C}-$$

\subsection{Trimming} \label{trim}
After Processing is finished, some of temporary extents of new/remaining intervals need to be corrected. This editing process is based on a similar one called {\em trimming} in the MMP algorithm, but slightly modified. For each pair of neighboring non-twin intervals sharing the same face created in Processing, we check whether there is an overlap or not; if so, we correct it and update their items $I.\Extent$.
Let $I-J$ be such a pair of intervals. We can divide into two possible cases according to whether $I.\Edge$ is equal to $J.\Edge$ or not. The former case is the same as described in the MMP algorithm, but the latter case is our original generalization. In the former case, we calculate the ridge point which hits the common edge of $I$ and $J$, and set it as the end point of $I$ and the starting point of $J$. In the latter case, we calculate the possible ridge point hitting each edge of $I$ and $J$. Namely, if the ridge point hits $I.\Edge$, set it as the end point of $I$ and set $J.\Extent$ to be empty, and if the ridge point hits $J.\Edge$, set it as the starting point of $J$ and set $I.\Extent$ to be empty.

In the process, it can happen that at least one of twin intervals has the empty extent.
We say that they are {\em redundant twin}.
By definition, each of twin intervals must have non-empty extent, thus we need to resolve the redundant twin.
If only one of their extents is empty, remove the interval, and if both extents are empty, remove one of them, e.g., let the first interval remain (notice that any enriched interval with the empty extent is still in use in the expression of the wavefront).
Finally we produce the new interval loop.

\section{Computational Complexity}
\subsection{Theoretical Upperbound}\label{complexity}
We assume that the source point $p$ is generic and the wavefront collapses to the farthest point without any bifurcation events during the propagation.
Let $n$ be the number of vertices of $S$. Then the number $b$ of (undirected) edges is $3(n-2)$ and the number $c$ of faces is $2(n-2)$;
indeed, we have $2b=3c$ and $n-b+c=2$ (the Euler characteristic of the $2$-sphere).

\begin{lemma}
	The number of the vertex events is $n$. The number of the collision events is $n - 2$ (here the final event is counted as one collision event).
\end{lemma}

\begin{proof}
	The number of ridge point on the wavefront increases by one at a vertex event, decreases by one at a collision event except the final event, and decreases by three in the final event. Other types of events do not change the number.
\end{proof}

\begin{lemma}
	At any given time $r$, the number of the ridge points is $O(n)$.
\end{lemma}
\begin{proof}
	The number of the ridge points in $W(r)$ is equal to or less than the number of vertex events happened, so it is $O(n)$.
\end{proof}

\begin{lemma}\label{edgeevents}
	The sum of the numbers of the edge (cross/swap) events is $O(n^2)$.
\end{lemma}
\begin{proof}
	It is equal to how many times the cut-locus $C$ intersects the edges. For each edge $e$, let $a$ an intersection point of $C$ and $e$. $a$ determines a sub-tree of $C$, which consists of the ridge points going to $a$. Those sub-trees are disjoint, therefore the number of possible intersection $C\cap e$ is at most $n$ (since every ridge originated from at least one vertex). Thus the number of all intersections is bounded by $nb$, therefore $O(n^2)$.
\end{proof}

\begin{lemma}
	All vertex events take $O(n)$ time in total to be processed. Each of other events takes $O(1)$ time per event to be processed.
\end{lemma}
\begin{proof}
	The total number of calculation caused by all vertex events is estimated to be $2b\, (=6(n-2))$, which is the number of directed edges.
	A cross event or swap event takes $O(1)$ time, because it has at most two intervals to be propagated or deleted.
	A collision event takes $O(1)$ time, because it has only one interval to be deleted and two adjacent intervals to be updated.
\end{proof}

\begin{theorem}
	Our algorithm takes $O(n^2 \log n)$ time and $O(n)$ space.
\end{theorem}
\begin{proof}
	While each event takes $O(1)$ time on average to be processed, it requires $O(\log n)$ time to be scheduled using a priority queue. The overall number of the events is $O(n^2)$, thus the time complexity is $O(n^2 \log n)$.
	There are $O(n)$ intervals in the interval loop at any given time. Because at most one event is associated with an interval, there are $O(n)$ events in the event queue at any given time. thus the space complexity is $O(n)$.
\end{proof}

\

Our algorithm can support the shortest path query using extra space complexity:

\begin{theorem}
	Our algorithm takes $O(n^2 \log n)$ time and $O(n^2)$ space for supporting the shortest path query.
\end{theorem}

\begin{proof}
	To do this, all intervals that are generated and removed from the wavefront during the algorithm running are required to the path query. They must be retained in the memory. For each edge, intervals are separated by the intersections with the cut-locus. Therefore there are $O(n^2)$ intervals overall and the space complexity is $O(n^2)$. The time complexity does not change.
\end{proof}

\subsection{Experimental Result}
An experimental result regarding computational complexity is shown below. We took the recursively subdivided surfaces of a regular octahedron using the Loop subdivision scheme \cite{Loop} (Figure \ref{subdivided_octahedron}).

\begin{figure}[h]
	\centering
	\includegraphics[height=6cm, pagebox=cropbox]{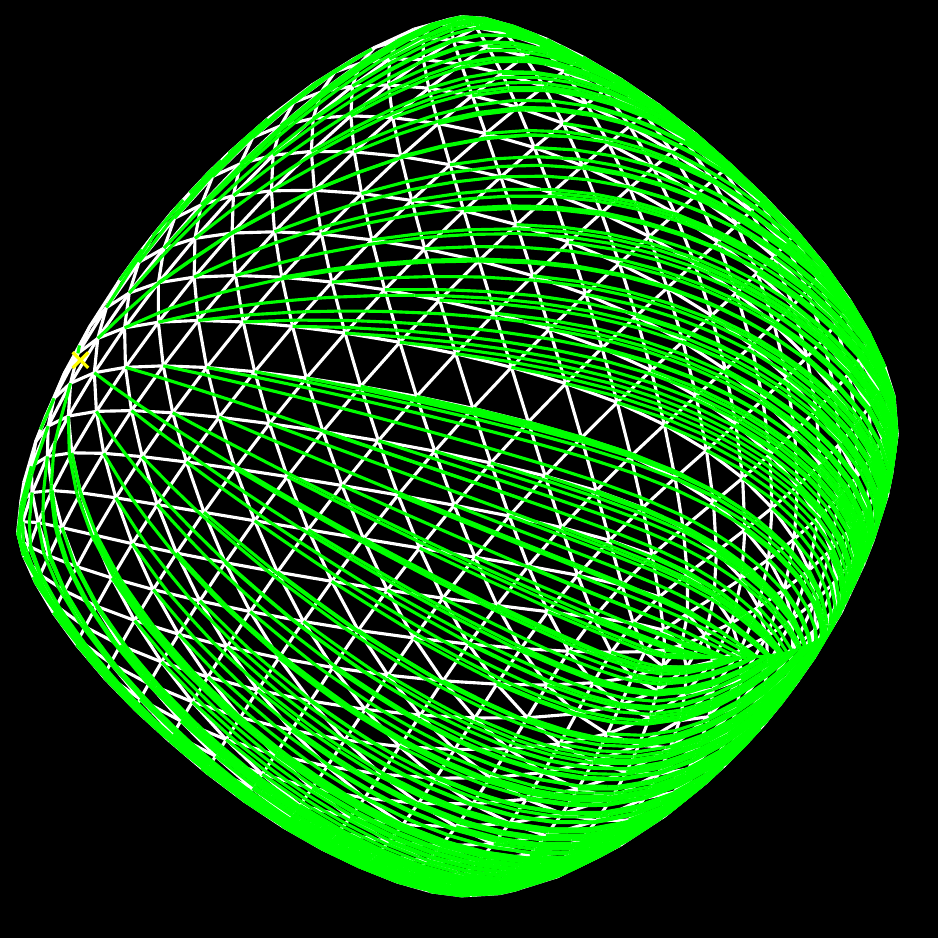}
	\includegraphics[height=6cm, pagebox=cropbox]{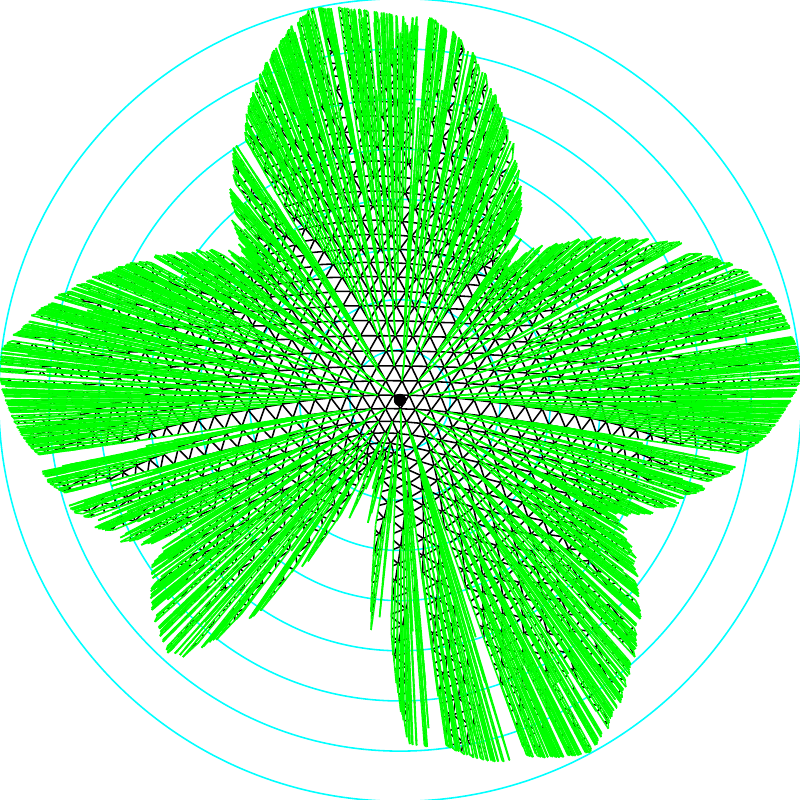}
	\caption{Level-4 subdivided surface of an octahedron and its source unfolding}
	\label{subdivided_octahedron}
\end{figure}

The table below shows the \textit{level} (how many times the subdivision performed from the initial octahedron), the number of vertices, the number of faces, the overall computation time, the memory usage, the number of total processed events, and the computation time per event. Here the $O(n^2)$ space variant (supporting the shortest path query) is used.

\begin{center}
\begin{tabular}{|r|r|r|r|r|r|r|}
level & vertices & faces & time (sec.) & memory (MB) & events & $\mu$s/event \\
4 & 1026 & 2048 & 0.051 & 11 & 15737 & 3.2 \\
5 & 4096 & 8192 & 0.506 & 71 & 125350 & 4.0 \\
6 & 16386 & 32768 & 4.315 & 481 & 889247 & 4.8 \\
7 & 65538 & 131072 & 40.214 & 3788 & 7158370 & 5.6
\end{tabular}
\end{center}

Like an experimental result of the MMP algorithm by Surazhsky et al. \cite{Surazhsky}, experimental performance of our algorithm is sub-quadratic, both in terms of time and space. This is due to the fact that estimation of the number of edge events given by Lemma \ref{edgeevents} is too pessimistic and sub-quadratic in practice. Note that the memory usage is measured by the runtime, and in fact, contains the 3D mesh data as well as other miscellaneous things that make our program actually works. We can see that the computation time per event is clearly linearly correlating with $\log n$, and considering that, an experimentally-estimated complexity in this example is given by $O(n^{1.47} \log n)$ time and $O(n^{1.47})$ space, since the number of events is estimated to be $O(n^{1.47})$.

\section{Conclusion}

In the present paper, we have proposed a novel generalization of the MMP algorithm; it produces an interactive visualization of the wavefront propagation and the cut-locus on a convex polyhedral surface $S$, and finally provides a nice planar unfolding of $S$ without any overlap, instantaneously and accurately.
Here we consider {\em generic} source points, that is sufficient for our practical purpose, and indeed, that enables us to classify what kind of geometric events arises in the wavefront propagation and makes the algorithm simple enough to be treated.
A main idea is to introduce the notion of an {\em interval loop}, which is a new data structure representing the wavefront. It is propagated as the time (distance) $r$ increases.
The computational complexity of our algorithm is the same as the original MMP, while our actual use is supposed for polyhedra with a reasonable size of number $n$ of vertices. We have successfully implemented our algorithm to computer -- it works well as expected and we have demonstrated a number of outputs. There is still large room for further development.

\section*{Acknowledgements}
The authors sincerely appreciate Professors Takashi Horiyama and Jin-ichi Ito for patiently listening to the first author's early studies and giving him valuable advices. This work was partly supported by GiCORE-GSB in Department of Information Science and Technology, Hokkaido University, and JSPS KAKENHI Grant Numbers JP18K18714.

%%%%%%%%%%%%%%%%%%%%%%%%%%%%%%%%%%%%%%%%%%%

\appendix
\section{}
Lemma \ref{lem1} is easy from the definition of generic geometric events.
We prove Lemma \ref{lem2} below.
Let $S$ be a convex polyhedral surface in $\R^3$.
Our task is to precisely characterize the set $\Gamma$ of non-generic source points on $S$. First, according to Definition \ref{generic}, the set $\Gamma$ consists of $p \in S$ where
\begin{itemize}
\item $p$ is a vertex or lies on an edge of $S$;
\item $p$ lies on the interior of a face such that at least one of the following properties holds: during the wavefront propagation $W(r, p)$ as $r$ varies,
\begin{enumerate}
\item[(v)] a ridge point passes through a vertex $v$ of $S$;
\item[(e)] a collision event happens at a point of an edge $e$ of $S$;
\item[(c)] more than two ridge points collide at once at a single point on a face of $S$, except for the case of final events;
\end{enumerate}
\end{itemize}

\noindent
\underline{Case (v)}: Let $p$ be an interior point of a face.
The case (v) means that there are multiple shortest paths between $p$ and $v$, and
that is equivalent to that $p$ lies on the cut-loci $C(S, v)$.
Thus $\Gamma$ contains the union of $C(S, v)$ for all vertices $v$.
As a remark, the case that a ridge point of $W(r, p)$ moves along an edge $e$ of $S$
is also regarded as the case (v), because the ridge point passes through at least one of endpoints (vertices) of $e$.

\

\noindent
\underline{Cases (e, c)}:
Let $p_0$ be an interior point of a face $\sigma$.
Suppose that a collision event for the wavefront $W(r,p_0)$ happens at $q_0 \in S$.
There are at least three shortest paths from $p_0$ to $q_0$.
Take an unfolding $\mathcal{S} \subset \R^2$ of part of $S$ so that each of three shortest paths from $p_0$ to $q_0$ is expressed by the line segment between $q_0$ and each of three different centers corresponding to $p_0$ (see Figure \ref{vc}, the right).
Pick a small disk $U$ centered at $p_0$ in the interior of $\sigma$.
Let $(s, t)$ be any linear coordinates of $U$ and $(x, y)$ the coordinates of $\R^2$ containing $\mathcal{S}$.
On $\mathcal{S}$, there are three copies of $U$;
to every $p \in U$, we assign three points on $\mathcal{S}$, say $p_A$, $p_B$ and $p_C$ ordered clockwise with respect to $q_0$.
By the construction of $\mathcal{S}$, coordinates of $p_A$, $p_B$ and $p_C$ are written by linear functions in $s, t$. Let $L_{AB}$ (resp. $L_{BC}$) be the vertical bisector of the segment between $p_B$ and $p_A$ (resp. $p_C$);
each bisector is written in the form $y=\frac{\beta x + \gamma}{\alpha}$, where $\alpha$, $\beta$ are linear functions and $\gamma$ is a quadratic function in $s, t$ (use a rotation of $\R^2$ if needed).
Compute the common point $q$ of $L_{AB}$ and $L_{BC}$, then we obtain $(x, y)=(\varphi_1(s, t), \varphi_2(s, t))$, where $\varphi_1$ and $\varphi_2$ are some rational functions in $s, t$. Since the initial solution exists for $p_0$, that is $q_0$, we may take $U$ sufficiently small so that the solution $q$ always exists for any $p \in U$. Then we find consecutive arcs $A-B-C$ participating in the wavefront caused from $p$, centered at $p_A$, $p_B$ and $p_C$, respectively, that meet a collision event at $q$ near $q_0$.

\begin{itemize}
\item[(e)]
Suppose that $q_0$ lies on an edge $e$. On $\mathcal{S}$, let $e$ be presented by part of the line given by a linear equation in $x, y$. Substitute $(x, y)$ by $(\varphi_1(s, t), \varphi_2(s, t))$, that yields an algebraic equation in $s, t$.
\item[(c)]
Suppose that more than two ridge points of $W(r, p_0)$ meet a collision event at $q_0$. Take a suitable unfolding of (part of) $S$ and consider four consecutive arcs $A-B-C-D$ in order. Three centers $p_A$, $p_B$ and $p_C$ define a rational map $q=(\varphi_1(s, t), \varphi_2(s, t))$, and the center $p_D$ of the last arc $D$ is constrained by the condition $d(p_D, q)=d(p_A, q)$. Hence, we have an algebraic equation in $s, t$.
\end{itemize}

Consequently, combining (v) as well, we see that for any point $p_0 \in \Gamma$, there is a neighborhood $U$ of $p_0$ in $S$ such that $\Gamma \cap U$ consists of finitely many algebraic curves in $U$.
Since $\Gamma$ is bounded and closed, it is covered by finitely many such open sets.
In particular, the complement $S-\Gamma$, the set of {\em generic} source points, is open and dense.
This completes the proof. \hfill $\square$

\end{document}